\documentclass{amsart}
\usepackage[utf8]{inputenc}

\usepackage{amssymb}
\usepackage{amsmath}
\usepackage{amsthm}
\usepackage[alphabetic]{amsrefs}
\usepackage{mathrsfs}
\usepackage{comment}
\usepackage[colorlinks]{hyperref}
\usepackage[hyphenbreaks]{breakurl} 
\usepackage{soul}
\usepackage{amscd} 
\usepackage{fullpage}
\usepackage[all]{xy} 
\usepackage{setspace}
\usepackage{mathtools}
\usepackage{enumitem}
\usepackage[normalem]{ulem}

\usepackage{latexsym}
\usepackage{graphicx}
\usepackage{setspace}
\usepackage{tikz}
\usepackage{tikz-cd}
\usepackage{pgfplots}
\usetikzlibrary{decorations.pathreplacing,angles,quotes}
\usetikzlibrary {decorations.markings} 
\usetikzlibrary{arrows,automata}
\usetikzlibrary{positioning}

\pgfplotsset{compat=1.18}

\usepackage{xcolor}

\newtheorem{theorem}{Theorem}[section]
\newtheorem{proposition}[theorem]{Proposition}
\newtheorem{lemma}[theorem]{Lemma}
\newtheorem{corollary}[theorem]{Corollary}

\newtheorem{observation}[theorem]{Observation}

\newtheorem{claim*}{Claim}

\theoremstyle{remark}
\newtheorem{example}[theorem]{Example}
\newtheorem{remark}[theorem]{Remark}
\newtheorem{remarks}[theorem]{Remarks}

\newcommand{\Aff}{{\mathbb A}}

\newcommand{\PP}{{\mathbb P}}

\newcommand{\FF}{{\mathbb F}}

\newcommand{\ZZ}{{\mathbb Z}}
\newcommand{\NN}{{\mathbb N}}

\newcommand{\calC}{{\mathcal C}}

\newcommand{\calP}{{\mathcal P}}

\newcommand{\calS}{{\mathcal S}}

\newcommand{\OO}{{\mathcal O}}

\newcommand{\hideqed}{\renewcommand{\qed}{}}

\DeclareMathOperator{\im}{im}

\DeclareMathOperator{\ev}{ev}

\DeclareMathOperator{\opt}{opt}

\numberwithin{equation}{section}
\numberwithin{table}{section}

\usepackage{color}

\newcommand{\defi}[1]{\emph{#1}} 

\title{Locally recoverable algebro-geometric codes \\
from projective bundles}

\author{Konrad Aguilar}
\address{Konrad Aguilar. Department of Mathematics and Statistics, Pomona College, 610 N. College Ave., Claremont, CA 91711, USA}
\email{{konrad.aguilar@pomona.edu}}

\author{Angelynn \'Alvarez}
\address{Angelynn \'Alvarez. Department of Mathematics, Embry-Riddle Aeronautical University, 3700 Willow Creek Rd., Prescott, AZ 86301, USA}
\email{{alvara44@erau.edu}}

\author{Ren\'e Ardila}
\address{Rene Ardila. Department of Mathematics, Grand Valley State University, 1 Campus Dr, Allendale, Michigan 49401, USA}
\email{{ardilar@gvsu.edu}}

\author{Pablo S. Ocal}
\address{Pablo S. Ocal. Okinawa Institute of Science and Technology, 1919-1 Tancha, Onna-son, Kunigami-gun, Okinawa 904-0495, Japan}
\email{{pablo.ocal@oist.jp}}

\author{Cristian Rodriguez Avila}
\address{Cristian Rodriguez Avila. Department of Mathematics and Statistics, Mount Holyoke College, 50 College Street, South Hadley, MA 01075, USA}
\email{{crodriguezavila@mtholyoke.edu}}

\author{Anthony V\'arilly-Alvarado}
\address{Anthony V\'arilly-Alvarado. Department of Mathematics MS 136, Rice University, 6100 S.\ Main St., Houston, TX 77005, USA}
\email{{av15@rice.edu}}

\date{March 18, 2025}



\begin{document}

\begin{abstract}
A code is locally recoverable when each symbol in one of its code words can be reconstructed as a function of $r$ other symbols. We use bundles of projective spaces over a line to construct locally recoverable codes with availability; that is, evaluation codes where each code word symbol can be reconstructed from several disjoint sets of other symbols. The simplest case, where the code's underlying variety is a plane, exhibits noteworthy properties: When $r = 1$, $2$, $3$, they are optimal; when $r \geq 4$, they are optimal with probability approaching $1$ as the alphabet size grows. Additionally, their information rate is close to the theoretical limit. In higher dimensions, our codes form a family of asymptotically good codes.
\end{abstract}

\maketitle

\section{Introduction}

Distributed cloud storage applications have long motivated the study of locally recoverable 
codes (LRCs), whose use 
has led to increased efficiency in both storage and data availability. 
For example, Meta uses an in-house implementation of Reed--Solomon codes~\cite{Meta24}, 
and many large-scale systems such as Windows Azure Storage \cite{HSXOCGLY12}, Hadoop \cite{SAPDVCB13}, and Facebook \cites{MLRHLLPSSTK14, MAK17} benefit from LRCs. 
This paper provides a practical construction of optimal LRCs and a general construction of asymptotically good LRCs with availability, leveraging ideas from algebraic geometry. Our codes have parameters with desirable properties for applications to cloud storage: Their minimum distance is high, enabling the correction of many errors; their information rate is close to the theoretical limit, implying a minimum amount of redundancy and overhead; and they allow multiple recovery sets, increasing the availability of the data for users while minimizing bandwidth usage.

\subsection{Algebro-geometric context}

While error-correcting codes date back to Hamming's work in the early 1950s~\cite{Hamming50}, the infusion of algebro-geometric techniques to create codes emerged only in 1977 with Goppa's construction of evaluation codes on curves that used the Riemann--Roch Theorem to bound their minimum distance~\cite{Goppa77}. Goppa elaborated on his idea in~\cite{Goppa81}, but it was only after Tsfasman, Vl{\u a}du{\c t}, and Zink~\cite{TsfasmanVladutsZink} showed how to use modular curves to beat the Gilbert--Varshamov bound that algebro-geometric methods took on a more central role in the development of codes with good theoretical properties~\cites{TsfasmanVluadact91, Walker00, HaymakerLopezMalmskogMatthewsPinero24}.

With the explosion of distributed large-scale storage in the early 2000s, a need arose for codes that could correct transmission errors \emph{and} repair data erasures, which led to the development of locally recoverable codes~\cites{HuangChenLi07, HanLastras-Montano07, GopalanChengSimitciYekhanin12, PapailiopoulosDimakis14}. In a seminal paper, Tamo and Barg~\cite{TamoBarg14} constructed LRCs whose minimum distance meets the Singleton-type bound that constrains LRCs. These codes inspired numerous algebro-geometric interpretations and further constructions, such as~\cites{BargTamoVluadact17, BargHaymakerHoweMatthewsVarillyAlvarado17, MunueraTenorio18, LiMaXing, MatthewsPinero20, MunueraTenorioTorres20, SalgadoVarillyAlvaradoVoloch21}, most of them relying on the structural geometry of certain curves (maybe embedded in a surface). 

A desirable layer of complexity one can add to LRCs is the property of 
\emph{availability}~\cites{WangZhang14, RawatPapailiopoulosDimakisVishwanath16}, whereby erasures 
in a code word can be repaired in multiple ways. Algebraic Geometry has also played a role in 
the construction of such codes; see for 
example~\cites{BargHaymakerHoweMatthewsVarillyAlvarado17, HaymakerMalmskogMatthews18, BartoliMontanucciMariaandQuoos20, JinKanZhang20, LopezMalmskogMatthewsPineroWooters21, CharaKottlerMalmskogThompsonWest23}. 

In 1972, Justesen~\cite{Justesen72} pioneered the systematic study of \emph{families} of codes with good asymptotic behavior; his own codes extended Reed--Solomon codes but did not use algebro-geometric techniques. Shortly after Goppa introduced algebro-geometric methods, Katsman, Tsfasman, and Vl{\u a}du{\c t} constructed algebro-geometric families of codes with good asymptotic properties~\cite{KatsmanTsfasmanVluadact84}. Further constructions appeared over time, e.g.~\cite{vLintSpringer87}, including recent work using algebraic surfaces~\cite{CouvreurLebacquePerret21}. It is natural to wonder if it's possible to construct algebro-geometric families of asymptotically good codes that incorporate locality and availability; to date, there are surprisingly few constructions of such families~\cites{BargTamoVluadact17, LiLiuMaXing24}.

\subsection{Results}

In this paper, we use products of affine and projective spaces to construct families of LRCs with availability that have good asymptotic properties. 
Our guiding examples are the following evaluation codes. Fix a finite field $\FF_q$, and positive integers $b$ and $r$, and pick a subset of $b (r+1)$ elements
\[
\calP = \{ (x_i,y_{i,j}) \}_{1 \leq i \leq b}^{1 \leq j \leq r+1} \subset \FF_q^2
\]
where all $x_i$ are distinct and all $y_{i,j}$ are distinct. Let
\[
V = \left\{ \sum_{\ell = 0}^{r-1} \sum_{i = 0}^{b-2} a_{ij} x^i y^{\ell}  : a_{ij}\in \FF_q \right\},
\]
a finite-dimensional $\FF_q$-vector space of polynomials. We construct a code $\calC$ by evaluating the points $\calP$ on the vector space $V$; see Section~\ref{ss:ag-codes} for details. Our first main result establishes values of $b$ and $r$ giving optimal locally recoverable codes.

\begin{theorem}[see Theorem~\ref{theo:line-fiber-optimal}]
Let $r = 1$, $2$, or $3$ and let $b$ satisfy $\displaystyle 3 \leq b \leq \frac{q}{r+1}$. The code $\mathcal{C}$ is optimal and locally recoverable with locality $r$. Its parameters $[n, k, d]_q$ are
\begin{align*}
n &= b (r+1),\\
k &= (b-1) r,\\
d &= r+3.
\end{align*}
\end{theorem}

This result is best possible in the following sense: For \emph{any} $q$, if $r \geq 4$, we can construct nonoptimal codes for all $b$ satisfying $\displaystyle 3 \leq b \leq \frac{q}{r+1}$. We do this concretely as a proof of concept in Section~\ref{sec:non-optimal} for $r = b = 4$ and $q = 37$. Knowing that not \emph{all} of our codes are optimal, we naturally ask \emph{what proportion} of these codes are optimal. Our second main result answers this question: When $q$ is large enough, for most choices of $\calP$, the code $\calC$ is optimal and locally recoverable. In other words, \emph{almost all} of the codes we construct are optimal. 

\begin{theorem}[see Theorem~\ref{theo:optimal-limit}]\label{theo:main-2}
Let $r \geq 4$ and let $b$ satisfy $\displaystyle 3 \leq b \leq \frac{q}{r+1}$. There exists an integer $q_0 = q_{0}(r,b)$ such that if $q \geq q_0$, for most choices of points $\calP$ there are no code words in $\calC$ of weight $\leq r + 2$. That is, the minimum distance of $\calC$ is $d \geq r+3$. Consequently, for most choices of points $\calP$ the code $\calC$ is optimal and locally recoverable with locality $r$. Moreover, as $q\to \infty$, the code $\calC$ is optimal with probability~$1$.
\end{theorem}

In all cases these locally recoverable codes have information rate
\[
\frac{k}{n} = \frac{b-1}{b}\cdot\frac{r}{r+1},
\]
approaching the theoretical limit of
\[
\frac{r}{r+1}
\]
when $r$ is fixed and $b$ is large. 

It is worth pointing out that the codes we construct to prove Theorem~\ref{theo:main-2} 
generalize the optimal LRCs of Tamo and Barg in~\cite{TamoBarg14}, while avoiding 
the use of good polynomials. A noteworthy difference between our codes 
and theirs is that we do not require evaluation points to be carefully chosen; a random choice of 
points yields an optimal LRC with probability~$\approx 1$, as long as 
$q$ is large (see Remark~\ref{rema:effective-crutch} for more details).

We then generalize these evaluation codes to algebro-geometric codes arising from bundles of projective spaces. Our third main result gives the parameters of these codes, as well as a lower bound on their minimum distance.

\begin{theorem}\label{theo:higher-dimensional-summary}
Let $\FF_q$ be a finite field of cardinality $q$. Fix positive integers $\alpha$, $\beta$, $b$, and $t$ such that $b \geq \alpha + 1$. There exists an integer $m_0 = m_0(\alpha,\beta,b,t,q)$ such that, for each $m\geq m_0$, there is a locally recoverable code over $\FF_q$ with locality $\displaystyle r = \binom{\beta + m}{\beta}$ and availability $t$, and parameters $[n,k,d]_q$ satisfying
\begin{align*}
n &= b\left(t r + 1\right),\\
k &= (\alpha + 1) r,\\
d &\geq (b - \alpha) \left((t-1) r + 2\right).
\end{align*}
In particular, these codes parametrized by $m \geq m_0$ form a family of asymptotically good codes.
\end{theorem}

\subsection*{Outline}

In Section~\ref{sec:preliminaries}, we establish the necessary background for the paper. 
In Section~\ref{sec:line-baseline-codes}, we use the affine plane to provide a 
construction of LRCs, we determine their parameters, we extensively study when and how they are 
optimal, and we showcase a method to increase their minimum distance. 
This includes a useful detour through matroid theory. 
In Section~\ref{sec:non-optimal}, we argue for the existence of nonoptimal LRCs among the 
aforementioned ones, and we construct one.
In Section~\ref{sec:construction}, we use arbitrarily high-dimensional projective varieties 
to generalize the construction of Section~\ref{sec:line-baseline-codes} to an infinite 
family of LRCs, we determine their parameters, and we show that it is a family of 
asymptotically good codes.

\subsection*{Notation}

The cardinality of a set $S$ will be denoted by $\#S$. 
A code will be denoted by a calligraphic letter $\calC$, 
and its dual code will be denoted by $\calC^{\bot}$. 
All the codes we handle are linear evaluation codes 
$\calC \coloneqq \mathrm{im} \ev_{\calP}$ with $\calP$ a finite set of points. 
A code word of $\calC$ will be denoted by a bold lowercase letter $\mathbf{c}$. 
Scalars in $\FF_q$ will be denoted by a lowercase $x$, and elements of 
$\FF_q^n$ for $n > 1$ will be denoted by a bold lowercase $\mathbf{x}$. 
The $n$-dimensional affine space over $\FF_q$ is denoted by $\Aff^n$, and its Cartesian 
coordinates are denoted by $\Aff_{x_1,\dots,x_n}^n$. The $n$-dimensional projective 
space over $\FF_q$ is denoted by $\PP^n$, its homogeneous coordinates are denoted 
by $\PP_{\mathbf{x}}^n$ where $\mathbf{x} = [x_0,\dots,x_n]$. When we write $\PP_{x_0,\dots,x_{n-1}}^n$, the coordinates $x_0,\dots,x_{n-1}$ are the affine coordinates in the chart $D_+(x_n)$ where the last homogeneous coordinate $x_n$ is nonzero. 
Given a projective variety $X$ over $\FF_q$, we denote by $X(\FF_q)$ the $\FF_q$-rational 
points of $X$.

\section{Preliminaries}\label{sec:preliminaries}

In this section we briefly recall the basic notions of locally recoverable codes 
with availability and the basic constructions of algebro-geometric codes. 
For the fundamentals we refer the reader to~\cites{TsfasmanVluadact91,Walker00}.

\subsection{Linear locally recoverable codes}\label{ss:linear-lrcs}

A \emph{linear code} $\calC$ over a finite field $\mathbb{F}_q$ is a linear subspace of $\mathbb{F}_q^n$. We call $n$ the \emph{length} of $\calC$. We denote by $k$ the \emph{dimension} of $\calC$ as an $\mathbb{F}_q$-vector space. We denote by $d$ the \emph{minimum distance} of $\calC$, which is the minimum pairwise separation between two distinct elements of $\calC$ in the Hamming metric, or equivalently the minimum Hamming weight of the nonzero code words of $\calC$, that is, the minimum number of nonzero coordinates of the nonzero code words of $\calC$. The \emph{information rate} of a linear code $\calC$ is the ratio $k/n$; the \emph{relative distance} is the ratio $d/n$.

A code $\calC$ is said to be \emph{locally recoverable} (LRC) with \emph{locality} $r$ if for each symbol $c_i$ in a code word $\mathbf{c} = (c_1,\dots,c_n) \in \calC$, there exists a \emph{recovery set} $R_i\subset \{1,\dots,n\}\setminus \{i\}$ with $\#R_i \leq r$ such that $c_i$ is a function of the symbols $\{c_j\}_{j \in R_i}$. In particular, if the $i$-th coordinate of $\mathbf{c}$ is lost, it can be recovered by accessing $\leq r$ other coordinates in $\mathbf{c}$. Trivially, every linear code is a LRC with locality $r = k$. 
Our convention of locality throughout this paper is also referred to as \emph{all-symbol locality} in the literature~\cite{BargTamoVluadact17}. An LRC $\calC$ with locality $r$ is said to have \emph{availability} $t$ if for every $c_i$ in a code word $\mathbf{c} = (c_1,\dots,c_n) \in \calC$ there exist $t$ disjoint recovery sets $R_{i,\ell}\subset \{1,\dots,n\}\setminus \{i\}$ with $\#R_{i,\ell} \leq r$ such that $c_i$ is a function of the symbols $\{c_j\}_{j \in R_{i,\ell}}$ for $1 \leq \ell \leq t$.

The parameters of an LRC $\calC$ are denoted $[n,k,d;r,t]_q$, or simply $[n,k,d;r]_q$ when $t = 1$. 
They are constrained by relations like the Singleton-type bound for the minimum distance
\begin{equation}\label{eq:singleton-type_bound}
       d \leq n - k - \bigg\lceil \frac{k}{r} \bigg\rceil + 2,
\end{equation}
and the following bound for the information rate
\begin{equation*}\label{eq:information_rate_bound_LRC}
    \frac{k}{n} \leq \frac{r}{r+1},
\end{equation*}
both of which were proven independently in~\cites{GopalanChengSimitciYekhanin12, PapailiopoulosDimakis14}.

An LRC $\calC$ with parameters $[n,k,d;r]_q$ whose minimum distance achieves equality in \eqref{eq:singleton-type_bound} is said to be \emph{optimal}. When an LRC $\calC$ with parameters $[n,k,d;r]_q$ has the property that each code word is partitioned into sets of $r+1$ elements where recoverability takes place within each set, as is often the case with algebro-geometric locally recoverable codes, a particularly simple proof of~\eqref{eq:singleton-type_bound} was given in~\cite{SalgadoVarillyAlvaradoVoloch21}*{Theorem~I.3}. 

Let $\mathscr{C} = \{\calC_i\}_{i=1}^{\infty}$ be a family of codes (not necessarily LRCs) and denote by $[n_i,k_i,d_i]_q$ the parameters of $\calC_i$ for all $i\geq 1$. We say that $\mathscr{C}$ is a family of \emph{asymptotically good} codes when
\[
\lim_{i\to\infty}\frac{d_i}{n_i} > 0,
\quad\text{and}\quad
\lim_{i\to\infty}\frac{k_i}{n_i} > 0.
\]

\subsection{Algebro-geometric evaluation codes}\label{ss:ag-codes}

Algebraic Geometry supplies an abundance of constructions for linear codes 
under the framework of \emph{evaluation codes}. To specify such a code one 
needs a triple $(X,\calP,V)$, where $X$ is a quasi-projective variety over 
a finite field $\FF_q$, $\calP = \{ P_1, \dots, P_n\}$ is a set of $n$ 
points in $X(\FF_q)$, and $V$ is a finite-dimensional subspace of the 
function field of $X$. From this data, we construct a linear code as 
the image of the \emph{evaluation map} 
\begin{equation*}
\begin{tikzcd}[row sep=3em]
\ev_\mathcal{P} \colon &[-3.2em] V \arrow{r} & \FF_q^n,\\[-3.1em]
{} & f \arrow[maps to]{r} & (f(P_1),\dots,f(P_n)).
\end{tikzcd}
\end{equation*}
Such an algebro-geometric code $\calC \coloneqq \mathrm{im} \ev_{\calP}$ has length $n$ and dimension
\[
k \coloneqq \dim_{\FF_q} \mathrm{im} \ev_{\calP}  = \dim_{\FF_q} V - \dim_{\FF_q} \ker \ev_{\calP}.
\]
In particular, when $\ev_{\calP}$ is injective we have $k = \dim_{\FF_q} V$.

Beyond constructing codes, Algebraic Geometry provides a toolbox for proving properties of the codes it furnishes. Classically, the Riemann--Roch Theorem gives bounds for the dimension and the minimum distance of a code arising from smooth projective curves (see~\cite{Goppa77} and~\cite{Walker00}*{Theorem~6.4}). 
In this paper, the cornerstone of several arguments to determine properties or bounds on the parameters of the LRCs we construct is the following remarkable and celebrated bound on the cardinality of the rational points on algebraic varieties over finite fields.

\begin{theorem}[Lang--Weil estimate~\cite{LangWeil54}*{Lemma 1}]\label{theo:lang-weil-estimate}
Let $m$, $d$, and $s$ be nonnegative integers with $d > 0$, and let $q$ be a prime power. There exists a positive constant $A(m,d,s)$ such that for every $\FF_q$ and every variety $X\subseteq \PP^m$ of pure dimension $s$ and degree $d$, we have
\begin{equation*}
\#X(\FF_q) \leq A(m,d,s) q^s.
\end{equation*}
\end{theorem}
Lang and Weil prove sharper results in their seminal paper~\cite{LangWeil54}. We only use the coarse result above in Example~\ref{ex:badcodes} and Theorem~\ref{theo:line-fiber-optimal} below because our arguments deal only with the asymptotic behavior of certain point counts.

\section{Codes from line bundles and certain baseline codes}\label{sec:line-baseline-codes}

In this section we construct LRCs that are optimal for low values of the locality parameter, and that are optimal with high probability for the remaining values of the locality parameter. As the size of the alphabet increases, these codes are optimal with probability $1$.

\subsection{Set-up}\label{ss:set-up}

We fix positive integers $b$, $r$, and a prime power $q$. We define $n \coloneqq b (r+1)$ and $k \coloneqq  (b-1) r$. 
Let $X \coloneqq \Aff^1_{x} \times \Aff^1_{y}$. Pick $b$ distinct points $x_1,\dots,x_b$ in $\Aff^1_x(\FF_q) = \FF_q$, which we 
call \emph{points on the base}, and pick $n$ distinct points
\[
y_{1,1},\dots,y_{1,r+1},\dots, y_{b,1},\dots,y_{b,r+1} \in \Aff^1_y(\FF_q) = \FF_q
\]
which together form the set of points
\[
\calP = \left\{ (x_i,y_{i,j}) \right\}_{1 \leq i \leq b}^{1 \leq j \leq r+1} \subset \FF_q^2 = X(\FF_q).
\]
This set-up is sketched in Figure~\ref{fig:code-schematics} in the case when $r=3$.

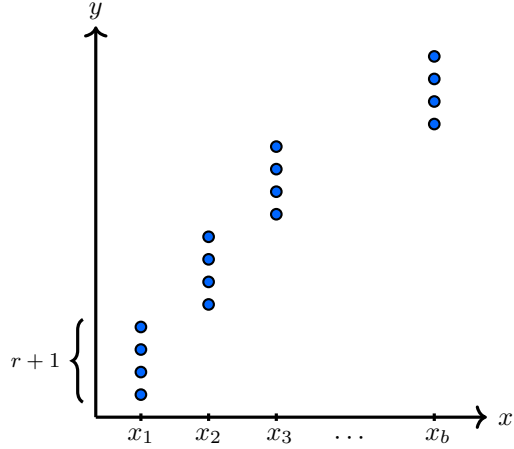
\begin{figure}[ht]
\begin{center}
\begin{tikzpicture}
\draw[->][very thick] (0,0) -- (5.2,0);
\draw[->][very thick] (0,0) -- (0,5.2);
\node [label={[xshift=0.25cm, yshift=-0.35cm]$x$}] at (5.2,0){};
\node [label={[xshift=0.cm, yshift=-0.15cm]$y$}] at (0,5.2){};

\filldraw [thick, fill=blue!60!cyan,opacity=0.99](0.6,0.3) circle [radius=2pt];
\filldraw [thick, fill=blue!60!cyan,opacity=0.99](0.6,0.6) circle [radius=2pt];
\filldraw [thick, fill=blue!60!cyan,opacity=0.99](0.6,0.9) circle [radius=2pt];
\filldraw [thick, fill=blue!60!cyan,opacity=0.99](0.6,1.2) circle [radius=2pt];

\filldraw [thick, fill=blue!60!cyan,opacity=0.99](1.5,1.5) circle [radius=2pt];
\filldraw [thick, fill=blue!60!cyan,opacity=0.99](1.5,1.8) circle [radius=2pt];
\filldraw [thick, fill=blue!60!cyan,opacity=0.99](1.5,2.1) circle [radius=2pt];
\filldraw [thick, fill=blue!60!cyan,opacity=0.99](1.5,2.4) circle [radius=2pt];

\filldraw [thick, fill=blue!60!cyan,opacity=0.99](2.4,2.7) circle [radius=2pt];
\filldraw [thick, fill=blue!60!cyan,opacity=0.99](2.4,3) circle [radius=2pt];
\filldraw [thick, fill=blue!60!cyan,opacity=0.99](2.4,3.3) circle [radius=2pt];
\filldraw [thick, fill=blue!60!cyan,opacity=0.99](2.4,3.6) circle [radius=2pt];

\filldraw [thick, fill=blue!60!cyan,opacity=0.99](4.5,3.9) circle [radius=2pt];
\filldraw [thick, fill=blue!60!cyan,opacity=0.99](4.5,4.2) circle [radius=2pt];
\filldraw [thick, fill=blue!60!cyan,opacity=0.99](4.5,4.5) circle [radius=2pt];
\filldraw [thick, fill=blue!60!cyan,opacity=0.99](4.5,4.8) circle [radius=2pt];

\draw[very thick] (0.6,-0.05) -- (0.6,0.05);
\draw[very thick] (1.5,-0.05) -- (1.5,0.05);
\draw[very thick] (2.4,-0.05) -- (2.4,0.05);
\draw[very thick] (4.5,-0.05) -- (4.5,0.05);

\begin{scope}
    \node [label={[xshift=0cm, yshift=-.6cm]$x_{1}$}] at (0.6,0){};
    \node [label={[xshift=0cm, yshift=-.6cm]$x_{2}$}] at (1.5,0){};
    \node [label={[xshift=0.05cm, yshift=-.6cm]$x_{3}$}] at (2.4,0){};
    \node [label={[xshift=0.05cm, yshift=-.6cm]$x_{b}$}] at (4.5,0){};
    \node [label={[xshift=0.05cm, yshift=-.6cm]$\cdots$}] at (3.35,0){};
\end{scope}

\draw [very thick][decorate,decoration={brace,amplitude=4pt,raise=2pt},yshift=0pt]
(-0.1,0.2) -- (-0.1,1.3) node [black,midway,xshift=-0.7cm] {\footnotesize $r+1$};

\end{tikzpicture}
\caption{The $n=b(r+1)$ points in $\mathcal{P}$.}
\label{fig:code-schematics}
\end{center}
\end{figure}

Label the points of $\mathcal{P}$ as $P_1,\dots P_n$. Consider the vector space of polynomials 
\begin{equation}\label{eq:vector-space-v}
V = \left\{\sum_{\ell = 0}^{r-1} a_\ell(x) y^\ell : a_\ell(x)\in \FF_q[x], \deg a_\ell(x) \leq b-2 \right\} \subset \FF_q[x,y],
\end{equation}
and define $\calC$ as the image of the linear evaluation map
\begin{equation}\label{eq:eval}
\begin{tikzcd}[row sep=3em]
\ev_\mathcal{P} \colon &[-3.5em] V \arrow{r} & \FF_q^n,\\[-3.1em]
{} & f(x,y) \arrow[maps to]{r} & (f(P_1),\dots,f(P_n)).
\end{tikzcd}
\end{equation}
It will be convenient to arrange the points $\calP$ in \emph{batches}
\[
A_i \coloneqq \{(x_i,y_{i,1}), \dots , (x_i,y_{i,r+1})\} \subset \calP
\]
for each $1 \leq i \leq b$, whence $\mathcal{P} = \coprod_{i=1}^b A_i$. The set $\{(x_i,y) : y\in \FF_q\} \subset X$ is called the \emph{fiber} above $x_i$ for $1\leq i\leq b$. A \defi{zero-fiber} of a polynomial $f(x,y)\in V$ is a batch $A_i$ where $f(P) = 0$ for 
all $P \in A_i$. A \defi{zero-fiber} of a code word $\mathbf{c} = \ev_{\calP}(f(x,y))$ 
for some $f(x,y)\in V$ is a zero-fiber of $f(x,y)$.

\newpage

\begin{remarks}\label{rema:simplicity-first-construction}
\
\begin{enumerate}[leftmargin=*]
\item For simplicity and clarity we begin by considering points in $\Aff^1_{x} \times \Aff^1_{y}$. We generalize this construction in Section~\ref{sec:construction}, where it is convenient to note that $\Aff_{x}^1 \times \Aff_{y}^1$ can be identified with an open subset of $\PP_{\mathbf{x}}^1 \times \PP_{\mathbf{y}}^1$.
\smallskip
\item The points in $\mathcal{P}$ are in \emph{general position} in the sense defined in Section~\ref{sec:construction}. 
Here, general position reduces to the statement that a nonzero polynomial $g(y) \in \FF_q[y]$ of degree $\leq r - 1$ cannot vanish along the $y$-coordinates of a batch of points. This follows by construction, since no two points in a batch $A_i$ share the same $y$-coordinate. 
This condition is crucial to our computation of the dimension of $\calC$ in Lemma~\ref{lem:injectivity}, our proof that the codes $\calC$ have locality $r$ in Lemma~\ref{lem:locality}, and our proof of optimality for small values of $r$ in Theorem~\ref{theo:line-fiber-optimal}.
\smallskip
\item \label{item:strong-general-position} The points in $\mathcal{P}$ satisfy a stronger condition than being in general position: No two share the same $y$-coordinate. We leverage this stronger condition in the probabilistic argument in Example~\ref{ex:badcodes} and Theorem~\ref{theo:optimal-limit}. These probabilistic arguments do not appear in Section~\ref{sec:construction}, which is why we use the weaker notion of general position there.
\smallskip
\item Our choice of polynomials in $V$ requires that the number $b$ of fibers satisfies $b \geq 2$. 
In Lemma~\ref{lem:exist_weight} below, we will require that $b \geq 3$ to find a code word in $\calC$ of weight $r+3$.
\end{enumerate}
\end{remarks}

\begin{example}\label{exam:small_r_optimal_example}
Let $q = 31$, $b = 4$ and $r = 3$, so that $n = 16$ and $k = 9$. Consider the set
\begin{align*}
    \mathcal{P} = \{ &(1,1),(1,2),(1,3),(1,4),\\
                     &(6,5),(6,6),(6,7),(6,8),\\
                     &(17,9),(17,10),(17,11),(17,12),\\
                     &(23,20),(23,21),(23,22),(23,23) \} \subset \FF^2_q,
\end{align*}
and the vector space of polynomials
\begin{align*}
V = \{ &(a_{00} + a_{01}x + a_{02}x^2) + (a_{10} + a_{11}x + a_{12}x^2)y\\
     + &(a_{20} + a_{21}x + a_{22}x^2)y^2 : a_{ij} \in \mathbb{F}_{31}\, \text{for all}\, 0\leq i,j\leq 2 \}.
\end{align*}
The code $\mathcal{C}$ is the image of the evaluation $\ev_{\mathcal{P}} \colon V \to \FF_{31}^{16}$ where $f(x,y) \mapsto (f(P_1),\dots,f(P_{16}))$. Let
\begin{align*}
f(x,y) &= (x-6)(x-23)(y-4)(y-10) \\
&= x^2y^2+17x^2y+9x^2+2xy^2+3xy+18x+14y^2+21y+2.
\end{align*}
The code word associated to $f(x,y)$ is
\[
\mathbf{c} = \ev_\calP(f(x,y)) = (25, 24, 26, 0, 0, 0, 0, 0, 20, 0, 3, 29, 0, 0, 0, 0).
\]
The sets of points
\[
A_2 \coloneqq \{(6,5),(6,6),(6,7),(6,8)\}\quad\text{and}\quad
A_3 \coloneqq \{(23,20),(23,21),(23,22),(23,23)\}
\]
are zero-fibers for $f(x,y)$ and (equivalently) for $\mathbf{c}$. The weight of $\mathbf{c}$ is $6$, showing that $d \leq 6$ for $\calC$. A {\tt Magma} or {\tt SageMath} calculation shows that this upper bound is sharp, that is, $d = 6$ for this code.
\end{example}

\subsection{Relations to existing codes}
\label{ss:relation-to-literature}

A code $\calC$ as in Section~\ref{ss:set-up} has important similarities and 
differences with the codes considered 
in~\cites{TamoBarg14,SalgadoVarillyAlvaradoVoloch21}. 
In~\cite{SalgadoVarillyAlvaradoVoloch21}*{Section III.B}, the authors consider 
algebro-geometric codes arising from a triple of data 
$(X,\calP,V[N])$ where $X = \Aff^1_x\times\Aff^1_y$, 
like our codes, and where 
\[
V[N] = \left\{\sum_{\ell = 0}^{r-1} a_\ell(x) y^\ell : a_\ell(x)\in \FF_q[x], \deg a_\ell(x) \leq N \right\} \subset \FF_q[x,y],
\]
for a nonnegative integer $N$. Our codes consider mostly the case where $N = b-2$, 
though we look at smaller values of $N$ in Section~\ref{ss:larger-minimum-distance}. 
The key difference between their codes and ours is the choice of points $\calP$ 
for code evaluation. We all consider a set of points $\calP$ partitioned into $b$ batches 
$A_1,\dots,A_b$ consisting of $r+1$ distinct points. 
However, in~\cite{SalgadoVarillyAlvaradoVoloch21}*{Section III.B} 
each batch is additionally 
constrained to satisfy an extra algebraic relation: 
If $A_i = \{(x_i,y_{i,j})\}_{j = 1}^{r+1}$, then there exists $g(x) \in \FF_q[x]$ 
a polynomial of degree $r+1$ such that $y_{i,j} = g(x_i)$ for all $1\leq i\leq b$ 
and $1\leq j\leq r+1$. 
Tamo and Barg's construction in~\cite{TamoBarg14} uses the polynomial 
$g(x) = x^{r+1}$, and more generally discusses the concept of 
\emph{good polynomials} $g(x)$. The polynomial $g(x)$ affords better control of the 
minimum distance of the resulting codes. Without this kind of control, we can only 
prove that the codes constructed in Section~\ref{ss:set-up} are optimal when 
$b \geq 3$ and $r = 1$, $2$ or $3$; see Theorem~\ref{theo:line-fiber-optimal}. 
However, when $r \geq 4$, our codes in Section~\ref{ss:set-up} are provably optimal 
with probability approaching $1$ for uniformly random choices of points 
$\calP$ as $q \to \infty$; see Theorem~\ref{theo:optimal-limit}. 
Although our codes for $r \geq 4$ are not always optimal 
(see Section~\ref{sec:non-optimal}), our investigation suggests that the use 
of good polynomials, like $g(x) = x^{r+1}$, imposes a serious constraint on 
the universe's supply of optimal LRCs.

Our codes also share superficial similarities with those of Munuera and Ten{\'o}rio~\cite{MunueraTenorio18}*{Section 2.2}, which is not surprising, as their codes generalize those in~\cites{TamoBarg14, BargTamoVluadact17}. However, our codes are neither a special case of the Munuera--Ten{\'o}rio construction, nor is there a clear common refinement of both constructions. To mimic our codes in the notation of~\cite{MunueraTenorio18}, one would need to take $m = 2$, $t = 1$, $\phi_1 = x_1 \eqqcolon x$, $\phi_2 = x_2 \eqqcolon y$, $\#\calS = b$, and $V_i = \{a(x) \in \FF_q[x] : \deg a(x) \leq b - 2\}$ for all $0 \leq i \leq r-1$. This would force $r = q-1$ in their set-up (a restriction we do not impose), and the set of evaluation points $\calP \subset \Aff^2_{x,y}(\FF_q)$ would contain $b(r+1) = bq$ points in $b$ batches with $r+1 = q$ overlapping $y$ coordinates. In our construction, it is essential that all the $b(r+1)$ points have distinct $y$ coordinates; see Remarks~\ref{rema:simplicity-first-construction}(\ref{item:strong-general-position}) and Figure~\ref{fig:code-schematics}.

\subsection{Dimension and locality of \texorpdfstring{$\calC$}{the code}}

\begin{lemma}\label{lem:injectivity}
The map $\ev_{\calP}$ in~\eqref{eq:eval} is injective. In particular, $\calC$ has dimension $k = (b-1) r$.
\end{lemma}

\begin{proof}
Suppose that $f(x,y) \in \ker(\ev_\mathcal{P})$ for $f(x,y) = \sum_{\ell=0}^{r-1}a_\ell(x) y^\ell$. We show that $f(x,y) \equiv 0$ in $\FF_q[x,y]$. As 
\[
\ev_{\mathcal{P}}(f(x,y)) = (0,\dots,0),
\] 
we have
\[
\sum_{\ell=0}^{r-1} a_\ell(x_i) y_{i,j}^\ell = f(x_i,y_{i,j}) = 0
\]
for all $1\leq i\leq b$ and all $1\leq j\leq r+1$. 
Hence, for all $1\leq i\leq b$ the polynomials $f(x_i,y) \in \FF_q[y]$ of degree $\leq r-1$ have the $r+1$ distinct zeros $y_{i,1},\dots,y_{i,r+1}$, implying $f(x_i,y) = 0$ as an element of $\FF_q[y]$. In turn, this shows that for all $0\leq \ell\leq r-1$ the polynomials $a_\ell(x) \in \FF_q[x]$ of degree $\leq b-2$ have the $b$ distinct zeros $x_1,\dots,x_b$, so $a_\ell(x) = 0$ as an element of $\FF_q[x]$. Hence $f(x,y)$ is the zero polynomial, as claimed, and
\begin{equation*}
\dim_{\FF_q} \calC \coloneqq \dim_{\FF_q} V = (b-1)r = k.
\eqno\qed
\end{equation*}
\hideqed
\end{proof}

\begin{lemma}\label{lem:locality}
The code $\mathcal{C}$ has locality $r$.
\end{lemma}

\begin{proof}
Let $f(x,y) = \sum_{\ell=0}^{r-1}a_\ell(x) y^\ell$ be an element of $V$, where $a_\ell(x)\in \FF_q[x]$ for all $0\leq \ell \leq r-1$. 
Suppose that a code word $\mathbf{c} = (f(P_1),\dots,f(P_n))$ is missing a symbol $f(P_i)$. 
Without loss of generality, we may assume that $i = 1$, so $P_1 = (x_{1},y_{1,1}) \in A_{1}$, 
where $A_{1} = \{(x_1,y_{1,j})\}_{1\leq j\leq r+1} \subset \mathcal{P}$. Consider the matrices 
\[
M=\begin{bmatrix}
    1 & y_{1,1} & y_{1,1}^2 & \cdots & y_{1,1}^{r-1}\\
    1 & y_{1,2} & y_{1,2}^2 & \cdots & y_{1,2}^{r-1}\\
    \vdots & \vdots & \vdots & \ddots & \vdots\\
    1 & y_{1,r} & y_{1,r}^2 & \cdots & y_{1,r}^{r-1}\\
    1 & y_{1,r+1} & y_{1,r+1}^2 & \cdots & y_{1,r+1}^{r-1}\\
\end{bmatrix}, \quad
\mathbf{a}=\begin{bmatrix}
    a_0(x_1)\\
    a_1(x_1)\\
    \vdots\\
    a_{r-1}(x_1)
\end{bmatrix},\quad
\text{and} \quad
\mathbf{F}=\begin{bmatrix}
    f(x_1,y_{1,1})\\
    f(x_1,y_{1,2})\\
    \vdots\\
    f(x_1,y_{1,r+1})
\end{bmatrix}.
\]
Note that $\mathbf{F} = M \mathbf{a} = \ev_{\mathcal{P}}|_{A_{1}}(f(x,y))$. 
That is, the components of $\mathbf{F}$ are the $r+1$ symbols in the code word obtained when restricting the evaluation to $A_1$. 
Let $M'$ and $\mathbf{F}'$ denote the matrices formed by deleting the first row in $M$ and $\mathbf{F}$, respectively. 
The matrix $M'$ is a square $r \times r$ Vandermonde matrix, with determinant
\[
\det(M') = \prod_{2\leq i < j \leq r+1} (y_{1,j}-y_{1,i}).
\]
Since no two points in $\mathcal{P}$ have the same $y$-coordinate, the matrix $M'$ has a nonzero determinant, and is thus invertible. Hence, we may solve $\mathbf{a} = (M')^{-1} \mathbf{F}'$. 
The components of $\mathbf{a}$ are the $r$ coefficients of the single variable polynomial $f(x_{1},y) \in \FF_q[y]$. The missing symbol is equal to the dot product of the row removed from $M$ and $\mathbf{a}$:
\[
f(P_1) = f(x_{1},y_{1,1}) = \sum_{\ell=0}^{r-1} a_{\ell}(x_{1})y_{1,1}^{\ell}.
\]
Therefore, 
we can recover any symbol of a code word using $r$ other symbols.
\end{proof}

\begin{corollary}
    \label{cor:singleton-type-bound}
    The minimum distance of $\calC$ satisfies $d \leq r + 3$.
\end{corollary}

\begin{proof}
    The code length of $\calC$ is $n = b(r+1)$ by construction. By Lemma~\ref{lem:injectivity}, $\calC$ has dimension $k = (b-1)r$. The Singleton-type bound~\eqref{eq:singleton-type_bound} gives the following upper bound for the minimum distance of $\calC$
\begin{equation*}
d \leq n - k - \left\lceil\frac{k}{r}\right\rceil + 2 = b(r+1) - (b-1)r - (b-1) + 2 = r + 3.
\eqno\qed
\end{equation*}
\hideqed
\end{proof}

The following lemma offers an alternative, constructive proof of Corollary~\ref{cor:singleton-type-bound} when $b\geq 3$.

\begin{lemma}\label{lem:exist_weight}
    Let $b \geq 3$. There exists a word in $\mathcal{C}$ that has weight $r+3$.
\end{lemma}

\begin{proof}
We exhibit a nonzero code word $\mathbf{c}$ with $r+3$ nonzero entries. As $b \geq 3$, we can define
\[
f(x,y) \coloneqq (x-x_1)(x-x_2)\cdots (x-x_{b-2})(y-y_{b-1,1})(y-y_{b-1,2})\cdots (y-y_{b-1,r-1}),
\]
and set $\mathbf{c} \coloneqq \ev_{\mathcal{P}}(f(x,y))$. The batches $A_1,\dots,A_{b-2}$ are zero-fibers for $f(x,y)$, so $f(x_i,y) \equiv 0$ in $\FF_q[y]$ for all $1\leq i\leq b-2$. We also have $f(x,y_{b-1,j}) \equiv 0$ in $\FF_q[x]$ for all $1\leq j\leq r-1$. 
Moreover 
\begin{align*}
&f(x_{b-1},y_{b-1,r}) \neq 0,\\
&f(x_{b-1},y_{b-1,r+1}) \neq 0,
\end{align*}
and
\[
f(x_b,y_{b,j}) \neq 0
\]
for $1\leq j\leq r+1$ because the $x$-coordinates are distinct, and the $y$-coordinates in the batch $A_{b-1}$ are distinct. 
Therefore $f(P) = 0$ for exacly $(b-2)(r+1)+(r-1)$ points $P\in \calP$. The number of nonzero entries of $\mathbf{c}$ is then 
\begin{equation*}
    n - ((b-2)(r+1) + (r-1)) = b(r+1) - (b-2)(r + 1) - (r-1) = r + 3.
    \eqno\qed
\end{equation*}
\hideqed
\end{proof}

\begin{example}\label{exam:small_r_optimal_example_bis}
The code word $\mathbf{c}$ in Example~\ref{exam:small_r_optimal_example} was constructed using the proof of Lemma~\ref{lem:exist_weight}.
\end{example}

\subsection{Minimum distance of \texorpdfstring{$\calC$}{the code} for small locality}

This subsection culminates in Theorem~\ref{theo:line-fiber-optimal}, where we show that, 
for small locality, the minimum distance achieves the Singleton-type bound in 
Corollary~\ref{cor:singleton-type-bound}. The constructed code $\calC$ is thus optimal.

We begin with an odd but remarkably useful observation.

\begin{observation}\label{obs:too-many-zeros}
If the set $A_i\subset \calP$ is a zero-fiber for a polynomial $f(x,y)\in V$, 
then $f(x_i,y)\in \FF_q[y]$ is a polynomial of degree $\leq r-1$ with $r+1$ zeros, 
whence $f(x_i,y) \equiv 0$ in $\FF_q[y]$. Consequently $f(x_i,y_{v,w}) = 0$ for all $1\leq v\leq b$ and all $1\leq w\leq r+1$.
\end{observation}

\begin{lemma}[Fiber Vanishing Lemma]\label{lem:max_number_zeros}
     Let $b \geq 3$ and let $f(x,y)\in V$ be a nonzero polynomial. Then $f(x,y)$ has $\leq b-2$ zero-fibers.
\end{lemma}

\begin{proof}
Let $f(x,y) = \sum_{\ell=0}^{r-1}a_\ell(x) y^\ell$ be an element of $V$, where $a_{\ell}(x)\in \FF_q[x]$ and $\deg a_{\ell}(x)\leq b-2$ for all $0\leq \ell\leq r-1$. 
Since $f(x,y) \not\equiv 0$ in $\FF_q[x,y]$, the code word $\mathbf{c} \coloneqq \mathrm{ev}_{\mathcal{P}}(f(x,y))$ is 
not zero by Lemma~\ref{lem:injectivity}. Without loss of generality, we may assume $f(P_n) \neq 0$ where $P_n = (x_b, y_{b,r+1})$. Then
\[
0 \neq f(P_n) = f(x_b,y_{b,r+1}) = \sum_{\ell=0}^{r-1} a_\ell(x_b) y_{b,r+1}^\ell
\]
so there exists an $m\in \{0,\dots,r-1\}$ such that $a_m(x_b) \neq 0$. In particular $a_m(x) \not\equiv 0$ in $\FF_q[x]$, so it has $\leq b-2$ zeros in $\FF_q$ because it has degree $\leq b-2$.

We now prove the claim by contradiction. Assume there exist $b-1$ values $i\in \{1, \dots, b\}$ such that $f(P) = 0$ for all $P\in A_i$. Since $f(x_b,y_{b,r+1}) \neq 0$ it must be that $f(x_i,y_{i,j}) = 0$ for all $1\leq i\leq b-1$ and all $1\leq j\leq r+1$. By Observation~\ref{obs:too-many-zeros} we have $f(x_i,y_{v,w}) = 0$ for all $1\leq i\leq b-1$, all $1\leq v\leq b$, and all $1\leq w\leq r+1$. For $1\leq i\leq b-1$ consider
\[
f(x_i,y) = \sum_{\ell=0}^{r-1} a_\ell(x_i) y^\ell \in \FF_q[y].
\]
These are polynomials of degree $\leq r-1$ with at least $\#\{y_{v,w}\}_{1\leq v \leq b}^{1\leq w\leq r+1} = b(r+1)$ zeros. Since $b(r+1) > r-1$ then $f(x_i,y) = 0$ as an element of $\FF_q[y]$ for all $1\leq i\leq b-1$. Thus the coefficients $a_\ell(x_i) = 0$ for all $1\leq i\leq b-1$ and all $0\leq \ell\leq r-1$. In particular $a_m(x)$ has $b-1$ zeros. This contradicts that $a_m(x)$ has $\leq b-2$ zeros, finishing the proof.
\end{proof}

We now establish an upper bound for the number of zeros of a given code word. We will subtract this from the length of the code $n$ to give a lower bound for the minimum distance. 

\begin{lemma}\label{lem:maximum-zeros-poly}
Let $b\geq 3$, let $f(x,y) \in V$ be a nonzero polynomial, 
and let $s\in \ZZ_{\geq 0}$ be the number of zero-fibers of $f(x,y)$. 
\begin{enumerate}
\item Then $f(P) = 0$ for $\leq b (r-1) + 2 s$ points $P\in \mathcal{P}$.

\item If $s = b-2$ then $f(P) = 0$ for $\leq (b-2)(r+1) + (r-1)$ points $P\in \mathcal{P}$.

\item Regardless of the value of $s$, if $r = 1, 2, 3$ then $f(P) = 0$ for $\leq (b-2)(r+1) + (r-1)$ points $P\in \mathcal{P}$.
\end{enumerate}
\end{lemma}

\begin{proof}
By Lemmas~\ref{lem:injectivity} and~\ref{lem:max_number_zeros}, since $f(x,y)$ is not zero, we have $0\leq s\leq b-2$. Fix $i \in \{1,\dots,b\}$. If there exists a point $P \in A_i$ such that $f(P) \neq 0$, since $f(x_i,y)\in \FF_q[y]$ has degree $\leq r-1$, then $f(x,y)$ can vanish on $\leq r-1$ points in $A_i$. Thus the number of points in $\mathcal{P}$ where $f(x,y)$ vanishes is at most
\[
s (r+1) + (b-s) (r-1) = b (r-1) + 2s,
\]
which increases as $s$ increases. If $s = b - 2$, then, without loss of generality $f(x,y)$ vanishes at $A_1\amalg \cdots \amalg A_{b-2}$. In particular, for $1\leq i\leq b-2$, the single-variable polynomial
\[
f(x_i,y) = \sum_{\ell=1}^{r-1}a_\ell(x_i)y^\ell \in \FF_q[y]
\]
of degree $\leq r-1$ vanishes at $y_{i,1},\dots,y_{i,r+1}$, which implies that $a_\ell(x)$ vanishes at $x_1,\dots,x_{b-2}$ for all $1\leq j\leq r+1$, whence
\[
a_\ell(x) = a'_\ell(x-x_1) (x-x_2) \cdots (x-x_{b-2})
\]
for some $a'_\ell \in \FF_q$, and thus we have the factorization
\[
f(x,y) = (x-x_1) (x-x_2) \cdots (x-x_{b-2}) p(y)
\]
for some polynomial $p(y)\in \FF_q[y]$ of degree $\leq r-1$. Since $(x-x_1)\cdots (x-x_{b-2}) \neq 0$ for $x \in \{x_{b-1},x_b\}$, the polynomial $f(x,y)$ can vanish on $\leq r-1$ points of $A_{b-1} \amalg A_{b}$. Thus, the number of points in $\mathcal{P}$ where $f(x,y)$ vanishes is at most
\[
(b-2)(r+1) + (r-1).
\]
The last part of the claim follows by observing that when $r = 1, 2, 3$ then for all $0\leq s\leq b - 3$ we have
\begin{equation}
\label{eq:r=123}
b (r-1) + 2s \leq (b-2)(r+1) + (r-1).
\hfill\qedhere
\end{equation}
\hideqed
\end{proof}

\begin{remark}
Note that, when $s = b - 2$, we have $(b-2)(r+1) + (r-1) \leq b (r-1) + 2 s$.
\end{remark}

Lemma~\ref{lem:maximum-zeros-poly} allows us to improve the bound in the Fiber 
Vanishing Lemma of zero-fibers for code words in $\calC$ whose weight is less 
than the Singleton bound $r+3$, as follows.

\begin{corollary}\label{coro:low-weight-zero-fibers}
Let $b \geq 3$ and $r\geq 4$. Let $\mathbf{c} \coloneqq \ev_{\calP}(f(x,y)) \in \calC$ be a nonzero code 
word of weight $\leq r + 2$. Then the number $s$ of zero-fibers of $\mathbf{c}$ satisfies $s \leq b-3$.
\end{corollary}

\begin{proof}
    We already know from the Fiber Vanishing Lemma that $s \leq b - 2$ for any nonzero code word $\mathbf{c}$.  Suppose that $s = b-2$. By Lemma~\ref{lem:maximum-zeros-poly}, $\mathbf{c}$ has weight at least
    \[
    n - ((b-2)(r+1) + (r-1)) = r + 3.
    \eqno\qed
    \]
    \hideqed
\end{proof}

Finally, we show that the code $\calC$ is optimal when the locality is small.

\begin{theorem}\label{theo:line-fiber-optimal}
Let $b \geq 3$ and $r = 1$, $2$, or $3$. Then the code $\mathcal{C}$ is an optimal LRC with parameters $[n,k,d;r]_q$, where
\begin{align*}
n &= b (r+1),\\
k &= (b-1) r,\\
d &= r+3.
\end{align*}
The code has information rate
\[
\frac{k}{n} = \frac{b-1}{b}\cdot \frac{r}{r+1}.
\]
\end{theorem}

\begin{proof}
The code $\calC$ has the claimed parameters $n$, $k$, and $r$ by~\eqref{eq:eval}, Lemma~\ref{lem:injectivity}, and Lemma~\ref{lem:locality}, respectively. In Corollary~\ref{cor:singleton-type-bound} we noted that the Singleton-type bound is $d\leq r+3$. By Lemma~\ref{lem:maximum-zeros-poly} a code word in $\mathcal{C}$ will have weight at least
\[
n - ((b-2)(r+1) + (r-1)) = b(r+1) - (b-2)(r + 1) - (r-1) = r+3
\]
whence $d\geq r+3$. Thus $d = r+3$ and $\calC$ is optimal because it reaches the Singleton-type bound.
\end{proof}

\subsection{Minimum distance of \texorpdfstring{$\calC$}{the code} for localities \texorpdfstring{$r \geq 4$}{greater than or equal to 4}}

The inequality~\eqref{eq:r=123} in the proof of Lemma~\ref{lem:maximum-zeros-poly} does not hold when $r \geq 4$. We thus lose control of the lower bound on the minimum distance of $\calC$, opening up the possible existence of nonoptimal codes $\calC$ when $r \geq 4$. We exhibit one such code in detail in Section~\ref{sec:non-optimal}. Nevertheless, we can still prove that, given a large alphabet, most choices of points for $\calC$ yield optimal LRCs.

Suppose that $\mathbf{c} \in \calC$ is a nonzero code word of weight $\leq r+2$, with 
$s$ zero-fibers; permuting indices if necessary, we may assume that the (disjoint) union of these 
zero fibers is $A_1\amalg \cdots \amalg A_s$. In particular, for $1\leq i\leq s$, 
the single-variable polynomial
$
    f(x_i,y) = \sum_{\ell=0}^{r-1}a_\ell(x_i)y^\ell \in \FF_q[y]
$
of degree $\leq r-1$ vanishes at $y_{i,1},\dots,y_{i,r+1}$. Thus $a_\ell(x)$ vanishes at $x_1,\dots,x_{s}$ for all $0\leq \ell \leq r-1$. Hence
\[
a_\ell(x) = g_\ell(x)(x-x_1) (x-x_2) \cdots (x-x_{s})
\]
for some $g_\ell(x) \in \FF_q[x]$ of degree $\leq (b - 2) - s$, so we have the factorization
\begin{equation}
\label{eq:specialf}
f(x,y) = (x-x_1) (x-x_2) \cdots (x-x_{s}) \sum_{\ell = 0}^{r-1} g_\ell(x) y^\ell.
\end{equation}
Now, the code word $\mathbf{c} = \ev_\calP(f(x,y))$ is supported on the fibers above $x_{s+1}, \dots,x_b$.  Since
\[
(x - x_1)\cdots(x - x_{s}) \neq 0
\]
for $x \in \{x_{s+1},\cdots,x_b\}$, a point $(x_i,y_{i,j})$ in the fibers above $x_{s+1}, \dots, x_b$ such that $f(x_i,y_{i,j}) = 0$ must satisfy
\[
\sum_{\ell = 0}^{r-1} g_\ell(x_i) y_{i,j}^\ell = 0.
\]
Thus, to have a code word of weight $\leq r + 2$, the polynomial 
\[
g(x,y) \coloneqq \sum_{\ell = 0}^{r-1} g_\ell(x) y^\ell \in \FF_q[x,y]
\]
must vanish on at least 
\[
(b-s)(r+1) - (r+2) = (b-s-1)r + (b - s - 2)
\]
points in the $b-s$ fibers over $x_{s+1}, \dots, x_b$. However, $g(x,y)$ has only $(b-s-1)r$ coefficients as a polynomial in $\FF_q[x,y]$. On the other hand, by Corollary~\ref{coro:low-weight-zero-fibers}, we know that $s \leq b-3$, so that $b - s - 2 \geq 1$. Thus, if we think of the coefficients of $g(x,y)$ as $(b-s-1)r$ unknowns satisfying $(b-s-1)r + (b - s - 2)$ linear relations of the form $g(x_i,y_{i,j}) = 0$, a code word $\mathbf{c} \in \calC$ of weight $\leq r+2$ solves a seemingly overconstrained linear system of equations. We should expect that, for most choices of points $\calP$ with distinct $x$- and $y$-coordinates, it is not possible to solve this system of linear equations. 
However, there may be polynomials of the form~\eqref{eq:specialf} with $s$ zero-fibers in multiple ways, and we must take into account polynomials whose zero-fibers are in arbitrary positions (not just over $x_1,\dots,x_{s}$). We can construct these by varying the roots of the largest factor of $f(x,y)$ that depends only on $x$ and changing the $g(x,y)$, while the collection of points $\calP$ remains \emph{fixed}. This gives new opportunities for $(b-s-1)r + (b - s - 2)$ of the remaining $(b-s)(r+1)$ points to interpolate the polynomial $g(x,y)$. 
More precisely, there are
\begin{equation}
\label{eq:badchancescount}
\sum_{s = 0}^{b-3}\binom{b}{s} \binom{(b-s)(r+1)}{(b-s-1)r + (b - s - 2)}
\end{equation}
choices of subsets in $\calP$ consisting of $s$ fibers and $(b-s-1)r + (b - s - 2)$ points in the remaining $(b-s)$ fibers. We must now estimate the probability that, given one such subset, there exists a nonzero polynomial $f(x,y) \in V$ vanishing along the subset, giving a code word of weight $\leq r + 2$. This leads to an estimate of the expected number of code words in $\calC$ of weight $\leq r + 2$.  To help fix ideas, we first illustrate this estimate in Example~\ref{ex:badcodes}.

\begin{example}\label{ex:badcodes}
    Let $b = 4$ and $r=4$, let $q > n = 20$ be a prime power, 
    and let $\calC$ be a code as in Section~\ref{ss:set-up}. 
    Suppose $f(x,y) = (x - x_1)g(x,y)$ gives rise to a code word of weight $\leq r + 2 = 6$, and that
    \[
    \{(x_1,y_{1,1}), (x_1,y_{1,2}), (x_1,y_{1,3}), (x_1,y_{1,4}),(x_1,y_{1,5})\}
    \]
    is the unique zero-fiber of $f(x,y)$, so $s = b-3 = 1$ in this case. Then the polynomial $g(x,y)$ takes the form
    \[
    \sum_{\ell = 0}^{3} g_\ell(x)y^\ell = (a_0 + a_1x) + (a_2 + a_3x)y + (a_4 + a_5x)y^2 + (a_6 + a_7x)y^3.
    \]
    If $g(x,y)$ were to pass through the $(b-s-1)r + (b - s - 2) = 9$ points 
    \begin{equation}
    \label{eq:coords}
    \begin{split}
    &(x_2,y_{2,1}), (x_2,y_{2,2}), (x_2,y_{2,3}), \\
    &(x_3,y_{3,1}), (x_3,y_{3,2}), (x_3,y_{3,3}), \\
    &(x_4,y_{4,1}), (x_4,y_{4,2}), (x_4,y_{4,3}),
    \end{split}
    \end{equation}
    in the remaining three fibers, we would have the following equality.
    \begin{equation} \label{eq:explicit-vdm}
    \begin{bmatrix}
        1 & x_2 & y_{2,1} & x_2y_{2,1} & y_{2,1}^2 & x_2y_{2,1}^2 & y_{2,1}^3 & x_2y_{2,1}^3\\
        1 & x_2 & y_{2,2} & x_2y_{2,2} & y_{2,2}^2 & x_2y_{2,2}^2 & y_{2,2}^3 & x_2y_{2,2}^3\\
        1 & x_2 & y_{2,3} & x_2y_{2,3} & y_{2,3}^2 & x_2y_{2,3}^2 & y_{2,3}^3 & x_2y_{2,3}^3\\
        1 & x_3 & y_{3,1} & x_3y_{3,1} & y_{3,1}^2 & x_3y_{3,1}^2 & y_{3,1}^3 & x_3y_{3,1}^3\\
        1 & x_3 & y_{3,2} & x_3y_{3,2} & y_{3,2}^2 & x_3y_{3,2}^2 & y_{3,2}^3 & x_3y_{3,2}^3\\
        1 & x_3 & y_{3,3} & x_3y_{3,3} & y_{3,3}^2 & x_3y_{3,3}^2 & y_{3,3}^3 & x_3y_{3,3}^3\\
        1 & x_4 & y_{4,1} & x_4y_{4,1} & y_{4,1}^2 & x_4y_{4,1}^2 & y_{4,1}^3 & x_4y_{4,1}^3\\
        1 & x_4 & y_{4,2} & x_4y_{4,2} & y_{4,2}^2 & x_4y_{4,2}^2 & y_{4,2}^3 & x_4y_{4,2}^3\\
        1 & x_4 & y_{4,3} & x_4y_{4,3} & y_{4,3}^2 & x_4y_{4,3}^2 & y_{4,3}^3 & x_4y_{4,3}^3
    \end{bmatrix}
    \begin{bmatrix}
        a_0\\
        a_1\\
        a_2\\
        a_3\\
        a_4\\
        a_5\\
        a_6\\
        a_7
    \end{bmatrix} = 
    \begin{bmatrix}
        0\\
        0\\
        0\\
        0\\
        0\\
        0\\
        0\\
        0\\
        0
    \end{bmatrix}
    \end{equation}
    This means that the above $9\times 8$ matrix's nine maximal minors vanish. Pick one such minor, say the one corresponding to the first row of the matrix, and consider it as a homogeneous polynomial in the $11$ variables 
    \[
    x_2,x_3,x_4,y_{2,2},y_{2,3},y_{3,1},y_{3,2},y_{3,3},y_{4,1},y_{4,2},y_{4,3}
    \]
    appearing in~\eqref{eq:coords}; the variable $y_{2,1}$ is missing because it appears only in the first row of the above matrix. These variables give rise to homogeneous coordinates in a projective space $\PP^{10}_{\FF_q}$.  The minor we selected defines a hypersurface $X \subset \PP^{10}_{\FF_q}$, a possibly reducible projective variety of dimension $9$ whose degree is independent of $q$. 
    By the Lang--Weil estimate (Theorem~\ref{theo:lang-weil-estimate}), there is a constant $A$, depending on the dimension and the degree of $X$, but not on $X$ itself, such that
    \[
    \# X(\FF_q) \leq A q^{9}.
    \]
    The space $\PP_{\FF_q}^{10}$ has $q^{10} + q^{9} + \cdots + q + 1$ rational points, and we would like to pick eleven coordinates $x_2,x_3,x_4$, $y_{2,1},\dots,y_{4,1}$ uniformly at 
    random. However, some care is required, because these eleven coordinates have restrictions; the $x_i$'s and the 
    $y_{i,j}$'s must be distinct. So, we must remove several hyperplanes from 
    $\PP^{10}_{\FF_q}$ before we draw our coordinates, such as the hyperplane 
    $x_2 = x_3$. There are $31$ such hyperplanes to be removed. 
    Each hyperplane is a 
    $\PP^{9}_{\FF_q}$, and thus contains only $q^{9} + q^{8} + \cdots + q + 1$ points. 
    The complement $U \subset \PP^{10}_{\FF_q}$ of all these hyperplanes has 
    $q^{10} + O(q^9)$ points as $q \to \infty$. Thus, 
    the probability that a point in $U$ chosen uniformly at random lies on $X$ is bounded above by
    \[
    \frac{\#X(\FF_q)}{\#U(\FF_q)} \leq \frac{Aq^9}{q^{10} + O(q^9)}\underset{q\to \infty}{\xrightarrow{\hspace*{2em}}} \frac{A}{q}.
    \]
    We deduce that, as $q\to\infty$, the expected number of code words in $\calC$ of weight $\leq 6 = r + 2$ is on the order of
    \begin{equation}
    \label{eq:prob_estimate}
    \left(\binom{4}{0} \binom{20}{15} + \binom{4}{1} \binom{15}{9}\right) \frac{A}{q}.
    \end{equation}
    This estimate is coarse. For example, it ignores the remaining eight minors. Nevertheless, it approaches $0$ as $q \to \infty$. Thus, as $q \to \infty$, the code $\calC$ will contain no words of weight $\leq 6$ with probability $1$.
\end{example}

We generalize Example~\ref{ex:badcodes} to prove one of our main results.

\begin{theorem}\label{theo:optimal-limit}
Let $b \geq 3$ and $r\geq 4$. There exists $q_0 = q_{0}(r,b)\in \NN$ such that if $q \geq q_0$, then 
for most choices of points $\calP$ there are no code words in $\calC$ of weight $\leq r + 2$. 
That is, the minimum distance of $\calC$ is $d \geq r+3$. Consequently, for most choices of 
points $\calP$ the code $\calC$ is optimal and locally recoverable with locality $r$. Moreover, 
as $q\to \infty$, choosing points $\calP$ uniformly at random yields an optimal code $\calC$ with probability $1$.
\end{theorem}

\begin{proof}
    Our discussion so far shows that if $f(x,y) \in V$ gives rise to a nonzero code word $\mathbf{c} \in \calC$ of weight $\leq r + 2$, then $f(x,y)$ vanishes along $s \leq b-3$ zero-fibers, and along at least $(b - s - 1)r + (b - s - 2)$ points in the remaining $b - s$ fibers. The expression~\eqref{eq:badchancescount} quantifies the number of subsets of $\calP$ along which such an $f(x,y)$ might vanish. Let $\calP'$ be one of these subsets, partitioned as $\calP' = \calP_1 \amalg \calP_2$, where $\calP_1 = A_{i_1} \amalg \cdots \amalg A_{i_s}$ are the $s$ zero-fibers above $x_{i_1},\dots,x_{i_s}$ and $\calP_2$ contains $(b - s - 1)r + (b - s - 2)$ points in the remaining $b-s$ fibers. Note that $f(x,y)$ is allowed to vanish at points of $\calP \setminus \calP'$, as long is it does not vanish along zero-fibers not already contained in $\calP_1$. Then
    \[
    f(x,y) = (x-x_{i_1}) (x-x_{i_2}) \cdots (x-x_{i_s}) g(x,y),
    \]
    where
    \[
    g(x,y) = \sum_{\ell = 0}^{r-1} g_\ell(x) y^\ell
    \]
    for some $g_\ell(x) \in \FF_q[x]$ of degree $\leq b - s - 2$ and $0 \leq \ell \leq r-1$. The number of distinct $x$-coordinates among the points in $\calP_2$ is $u \coloneqq b - s$, by definition of $s$, and the number of distinct $y$-coordinates is $v \coloneqq \#\calP_2 = (b - s - 1)r + (b - s - 2)$. Just as in~\eqref{eq:explicit-vdm}, the condition that $g(P) = 0$ for all $P \in \calP_2$ can be written as a matrix equation
    \begin{equation*}
    \label{eq:zero-force}
    M \mathbf{a} = 0,
    \end{equation*}
    where $M$ is a 
    \[
    ((b - s - 1)r + (b - s - 2)) \times (b - s -1)r
    \]
    matrix in $u+v = (b - s - 1)(r + 2)$ variables, and $\mathbf{a}$ encodes the coefficients of $g(x,y)$. If $f(x,y)$ is not the zero polynomial, the matrix $M$ must have a nontrivial kernel, which means that all its $\binom{(b - s - 1)r + (b - s - 2)}{(b - s - 2)}$ maximal minors 
    must vanish. Each minor is a homogeneous polynomial in $N \coloneqq u + v - (b - s - 2)$ variables, because each row removed from $M$ to obtain a maximal minor reduces the total number of variables by one. Thus, each maximal minor defines a hypersurface $X \subset \PP^{N-1}_{\FF_q}$. The Lang--Weil estimate (Theorem~\ref{theo:lang-weil-estimate}) implies that
    \[
    \#X(\FF_q) \leq A q^{N-2},
    \]
    where $A$ is a constant the depends on the degree of $X$ and on $N$, but not on $X$ itself. On the other hand,
    \[
    \#\PP^{N-1}(\FF_q) = \frac{q^N - 1}{q - 1} = q^{N-1} + q^{N-2} + \cdots + q + 1.
    \]
    We want to choose the points in $\calP$ uniformly at random, but we must be careful to choose distinct $x$- and $y$-coordinates. This means that among all points in $\PP^{N-1}(\FF_q)$, we must avoid hyperplanes of the form $x_i - x_j = 0$ for distinct $i,j \in \{1,\dots,n\}\setminus \{i_1,\dots,i_s\}$ and $y_{i,j} - y_{\ell,m} = 0$ 
    for distinct pairs of $y$-coordinates among the points in $\calP_2$. 
    Each such hyperplane is a $\PP^{N-2}_{\FF_q}$, and thus contains $q^{N-2} + \cdots + q + 1$ rational points. The total number $B$ of bad hyperplanes depends on $b$, $r$, and $s$, but not on $q$. Setting
    \[
    U \coloneqq \PP^{N-1}_{\FF_q} \setminus \{\text{bad hyperplanes}\},
    \]
    we deduce that 
    \[
    \#U(\FF_q) \geq \frac{q^N - 1}{q - 1} - B\cdot \frac{q^{N-1} - 1}{q - 1} = q^{N-1} + O(q^{N-2})
    \]
    as $q\to\infty$. Putting all this together, we see that if we select a point in $U$ uniformly at random, then the probability that one maximal minor of $M$ vanishes is
    \begin{equation}
    \label{eq:prob-estimate-general}
    \frac{\#X(\FF_q)}{\#U(\FF_q)} \leq \frac{Aq^{N-2}}{q^{N-1} + O(q^{N-2})} \underset{q\to \infty}{\xrightarrow{\hspace*{2em}}} \frac{A}{q}.
    \end{equation}
    Using the count~\eqref{eq:badchancescount}, the expected number of code words of weight $\leq r + 2$ is bounded above by
    \[
    \frac{\#X(\FF_q)}{\#U(\FF_q)} 
    \sum_{s = 0}^{b-3}\binom{b}{s} \binom{(b-s)(r+1)}{(b-s-1)r + (b - s - 2)}.
    \]
    Now~\eqref{eq:prob-estimate-general} guarantees that this quantity approaches $0$ as $q \to \infty$, which finishes the proof.
\end{proof}

\begin{remarks}\label{rema:effective-crutch}
\ 
\begin{enumerate}[leftmargin=*]
    \item It would be interesting to make the proof of Theorem~\ref{theo:optimal-limit} effective. That is, for a given pair of thresholds $0 < \gamma_1,\gamma_2 < 1$, obtain an explicit estimate of how large $q_0 = q_0(\gamma_1,\gamma_2)$ must be so that a proportion $\geq \gamma_1$ of the choices for $\calP$ yields an optimal LRC with probability $\geq \gamma_2$.
    \smallskip
    \item The optimal LRCs constructed by Tamo and Barg in~\cite{TamoBarg14} are special cases of the codes $\calC$ constructed in this section. They arise when points are chosen carefully to lie along the affine curve $y = g(x)$, where $g(x)$ is a \emph{good polynomial} such as $g(x) = x^{r+1}$. We refer the reader to~\cite{SalgadoVarillyAlvaradoVoloch21}*{Section III.A} for more details. In this context, Theorem~\ref{theo:optimal-limit} says that if one is willing to take $q$ very large, then with high probability one does not need to constrain the points $\calP$ to lie on such a curve to obtain an optimal LRC. As in~\cite{TamoBargFrolov16}, requiring a large alphabet is a mild restriction.
\end{enumerate}
\end{remarks}

\subsection{A detour through matroids}
\label{ss:matroids}

A natural question arising from our construction is whether the choice of points on the base alters the minimum distance of the resulting code. In this subsection, we review basic results on matroids and answer the question in the negative (see Theorem~\ref{theo:IsomMatroids}). 
Along the way we establish Corollary~\ref{coro:distance-iso-matroids}, a result of independent interest between matroids and codes.

A \emph{matroid} $\mathsf{M}$ is a pair $(M,\mathrm{rank}_{M})$ where $M$ is a finite set and $\mathrm{rank}_{M}:2^{M} \to \NN$ is a function of sets with
\begin{enumerate}
\item $0 \leq \mathrm{rank}_{M}(L) \leq \#L$ for all $L\subset M$.
\smallskip
\item $\mathrm{rank}_{M}(K) \leq \mathrm{rank}_{M}(L)$ for all $K\subset L\subset M$.
\smallskip
\item $\mathrm{rank}_{M}(K\cup L) + \mathrm{rank}_{M}(K\cap L) \leq \mathrm{rank}_{M}(K) + \mathrm{rank}_{M}(L)$ for all $K\subset L\subset M$.
\end{enumerate}
We say that a set $L\subset M$ is a \emph{circuit} when $\mathrm{rank}_{M}(L) \neq \#L$ 
and $\mathrm{rank}_{M}(K) = \#K$ for all $K\subsetneq L$. 
Two matroids $\mathsf{M} = (M,\mathrm{rank}_{M})$ and $\mathsf{N} = (N,\mathrm{rank}_{N})$ are 
\emph{isomorphic} if there exists a bijection of sets $\varphi\colon M\to N$ such that 
$\mathrm{rank}_{M}(L) = \#L$ if and only if $\mathrm{rank}_{N}(\varphi(L)) = \#\varphi(L)$ for all 
$L \subset M$. 
Given a matroid $\mathsf{M} = (M,\mathrm{rank}_{M})$, its \emph{dual} matroid 
$\mathsf{M}^{\bot} = (M^{\bot},\mathrm{rank}_{M^{\bot}})$ has $M^{\bot} \coloneqq M$ and 
$\mathrm{rank}_{M^{\bot}}(L) \coloneqq \#L + \mathrm{rank}_{M}(M\setminus L) - \mathrm{rank}_{M}(M)$ 
for all $L\subset M$.

\begin{remark}\label{rema:iso-circuit}
The rank of a circuit $L$ in a matroid $\mathsf{M}$ is $\#L-1$. Given $\varphi \colon \mathsf{M} \to \mathsf{N}$ an isomorphism of matroids, then $L$ is a circuit in $\mathsf{M}$ if and only if $\varphi(L)$ is a circuit in $\mathsf{N}$.
\end{remark}

The matroids we consider come exclusively from generator matrices of codes.

\begin{example}
Let $\calC$ be a code with generator matrix $G$, let $M$ be the set of columns of $G$, and let $\mathrm{rank}_{M}(L)$ be the rank of the submatrix of $G$ formed by the columns in $L\subset M$. Then $\mathsf{M}_{\calC} = (M,\mathrm{rank}_{M})$ is a matroid.
\end{example}

We will exploit the following observation.

\begin{lemma}[see~\cite{TamoPapailiopoulosDimakis13}*{Section III.B}]\label{lem:d-min-matroid}
The minimum distance of a code coincides with the cardinality of the smallest circuit in the matroid represented by its parity check matrix. \qed
\end{lemma}

\begin{corollary}\label{coro:distance-iso-matroids}
Let $\calC$ and $\calC'$ be codes with associated matroids $\mathsf{M}_{\calC}$ and $\mathsf{M}_{\calC'}$, respectively. If $\mathsf{M}_{\calC}$ and $\mathsf{M}_{\calC'}$ are isomorphic as matroids then the minimum distance of $\calC$ coincides with the minimum distance of $\calC'$.
\end{corollary}

\begin{proof}
Set $d_{\calC}$ and $d_{\calC'}$ the minimum distances of $\calC$ and $\calC'$, respectively. By Lemma~\ref{lem:d-min-matroid} there are circuits $L^{\bot}$ in $\mathsf{M}_{\calC^{\bot}}$ and $L'^{\bot}$ in $\mathsf{M}_{\calC'^{\bot}}$ such that $d_{\calC} = \#L^{\bot}$ and $d_{\calC'} = \#L'^{\bot}$. The isomorphism of matroids $\mathsf{M}_{\calC} \cong \mathsf{M}_{\calC'}$ induces an isomorphism between the dual matroids $(\mathsf{M}_{\calC})^{\bot} \cong (\mathsf{M}_{\calC'})^{\bot}$. 
This induces an isomorphism $\varphi : \mathsf{M}_{\calC^{\bot}} \cong (\mathsf{M}_{\calC})^{\bot} \cong (\mathsf{M}_{\calC'})^{\bot} \cong \mathsf{M}_{\calC'^{\bot}}$ where the first and third isomorphisms occur by \cite{JurriusPellikaan13}*{p.\ 269}. Thus $\varphi^{-1}(L'^{\bot})$ is a circuit in $\mathsf{M}_{\calC^{\bot}}$ and $\varphi(L^{\bot})$ is a circuit in $\mathsf{M}_{\calC'^{\bot}}$ by Remark~\ref{rema:iso-circuit}. Since $L^{\bot}$ and $L'^{\bot}$ have the smallest cardinality among the circuits in $\mathsf{M}_{\calC^{\bot}}$ and $\mathsf{M}_{\calC'^{\bot}}$ respectively, then $\#L^{\bot} \leq \#\varphi^{-1}(L'^{\bot}) = \#L'^{\bot}$ and $\#L'^{\bot} \leq \#\varphi(L^{\bot}) = \#L^{\bot}$. Thus $d_{\calC} = \#L^{\bot} = \#L'^{\bot} = d_{\calC'}$.
\end{proof}

We now consider two of the codes we constructed in Section~\ref{ss:set-up} that differ only in the choice of points on the base. Namely, fix positive integers $b$, $r$, and a prime power $q$, set $n = b(r+1)$, and set $\calC \coloneqq \mathrm{im} \ev_{\calP}$ and $\calC' \coloneqq \mathrm{im} \ev_{\calP'}$ for $\calP = \left\{ (x_i,y_{i,j}) \right\}_{1 \leq i \leq b}^{1 \leq j \leq r+1}$ and $\calP = \left\{ (x'_i,y_{i,j}) \right\}_{1 \leq i \leq b}^{1 \leq j \leq r+1}$ subsets of $\Aff_x^1(\FF_q)\times \Aff_y^1(\FF_q)$. We can relate Vandermonde-like matrices constructed from the points on the base of $\calC$ and $\calC'$.

\begin{lemma}\label{lem:KeyRowEquiv}
Consider the sets $S = \{x_1, \dots, x_b\}$ and $S' = \{x'_1, \dots, x'_{b}\}$. Denote by
\begin{equation*}
M = \begin{bmatrix}
  1 & 1 & \dots & 1\\
  x_1 & x_2 & \dots & x_b\\
  \vdots & \vdots & \ddots &\vdots\\
  x_1^{b-2} & x_2^{b-2} & \cdots & x_b^{b-2}
\end{bmatrix}\quad \text{and}\quad M' = \begin{bmatrix}
  1 & 1 & \dots & 1\\
  x'_{1} & x'_{2} & \dots & x'_{b}\\
  \vdots & \vdots & \ddots &\vdots\\
  (x'_{1})^{b-2} & (x'_{2})^{b-2} & \cdots & (x'_{b})^{b-2}
\end{bmatrix}
\end{equation*}
the $(b-1) \times b$ Vandermode-type matrices coming from $S$ and $S'$ respectively. There exists a $(b-1) \times (b-1)$ invertible matrix $A$ and a $b\times b$ invertible diagonal matrix $D$ such that $A M D = M'$.
\end{lemma}

\begin{proof}
The rank of $M$ and $M'$ is $b-1$, so each of their kernels is one dimensional. Set
\begin{equation*}
\mathbf{v} = \begin{bmatrix}
  \prod_{i\in \{1,\dots,b\}\setminus \{1\}} \frac{1}{x_1-x_i}\\
  \vdots\\
  \prod_{i\in \{1,\dots,b\}\setminus \{j\}} \frac{1}{x_j-x_i}\\
  \vdots\\
  \prod_{i\in \{1,\dots,b\}\setminus \{b\}} \frac{1}{x_b-x_i}
\end{bmatrix}\quad \text{and}\quad \mathbf{v}' = \begin{bmatrix}
  \prod_{i\in \{1,\dots,b\}\setminus \{1\}} \frac{1}{x'_1-x'_i}\\
  \vdots\\
  \prod_{i\in \{1,\dots,b\}\setminus \{j\}} \frac{1}{x'_j-x'_i}\\
  \vdots\\
  \prod_{i\in \{1,\dots,b\}\setminus \{b\}} \frac{1}{x'_b-x'_i}
\end{bmatrix},
\end{equation*}
a calculation yields $\ker M = \mathrm{span}(\mathbf{v})$ and $\ker M' = \mathrm{span}(\mathbf{v}')$. Set
\begin{equation*}
D = \mathrm{diag}\left[ \prod_{i\in \{1,\dots,b\}\setminus \{1\}} \frac{x'_1-x'_i}{x_1-x_i}, \dots, \prod_{i\in \{1,\dots,b\}\setminus \{j\}} \frac{x'_j-x'_i}{x_j-x_i}, \dots, \prod_{i\in \{1,\dots,b\}\setminus \{b\}} \frac{x_b-x_i}{x'_b-x'_i} \right]
\end{equation*}
so that $MD \mathbf{v}' = M \mathbf{v} = 0$. Thus $\ker MD = \ker M'$, so the row spaces of $MD$ and $M'$ coincide, so by doing elementary row operations to $MD$ we can reach $M'$. The matrix $A$ encoding those elementary row operations satisfies $AMD = M'$ as claimed.
\end{proof}

For $V$ as in~\eqref{eq:vector-space-v} fix the ordered basis
\[
\{1, x, \dots, x^{b-2}, y, yx, \dots, yx^{b-2}, \dots, y^{r-1}, y^{r-1}x, \dots, y^{r-1}x^{b-2}\}
\]
and write $G$ and $G'$ for the generator matrices of the codes $\calC$ and $\calC'$, respectively, with respect to this basis. Denote by $\mathsf{M}_{\calC}$ and $\mathsf{M}_{\calC'}$ their corresponding matroids.

\begin{theorem}\label{theo:IsomMatroids}
The matroids $\mathsf{M}_{\calC}$ and $\mathsf{M}_{\calC'}$ are isomorphic. Indexing the columns of $G$ and $G'$ in order by the elements $\{1,\dots,n\}$, the identity bijection $\mathrm{id}_{\{1,\dots,n\}} : \{1,\dots,n\} \to \{1,\dots,n\}$ gives a matroid isomorphism from $\mathsf{M}_{\calC}$ to $\mathsf{M}_{\calC'}$. In particular, the minimum distance of $\calC$ coincides with the minimum distance of $\calC'$.
\end{theorem}

\begin{proof}
    To show that the identity map on groundsets induces an isomorphism of matroids between 
    $\mathsf{M}_{\calC}$ and $\mathsf{M}_{\calC'}$ it suffices to show that we can get from 
    $G$ to $G'$ via a sequence of row operations and column scaling. 
    That is, we want to find an invertible $k\times k$ matrix $T$ and 
    a diagonal $n\times n$ matrix $R$ satisfying $T G R  = G'$. 
    Let $S = \{x_1, \dots, x_b\}$ and $S' = \{x'_1, \dots, x'_{b}\}$, 
    and let $A$ and $D = \mathrm{diag}[D_1, \dots, D_b]$ be as in Lemma~\ref{lem:KeyRowEquiv}. 
    Consider 
\begin{equation*}
T = \begin{bmatrix}
  A & \mathbf{0} & \dots & \mathbf{0}\\
  \mathbf{0} & A & \dots & \mathbf{0}\\
  \vdots & \vdots & \ddots & \vdots\\
  \mathbf{0} & \mathbf{0} & \dots & A
\end{bmatrix}\quad \text{and}\quad R = \begin{bmatrix}
  D_{1} I_{r+1} & \mathbf{0} & \dots & \mathbf{0}\\
  \mathbf{0} & D_{2} I_{r+1} & \dots & \mathbf{0}\\
  \vdots & \vdots & \ddots & \vdots\\
  \mathbf{0} & \mathbf{0} & \dots & D_{b} I_{r+1}
\end{bmatrix}
\end{equation*}
where $T$ is formed by $r^2$ blocks of size $(b-1) \times (b-1)$, and $I_{r+1}$ is the $(r+1) \times (r+1)$ identity matrix. A calculation confirms that $T G R  = G'$. Since $\mathsf{M}_{\calC} \cong \mathsf{M}_{\calC'}$, the final claim follows from Corollary~\ref{coro:distance-iso-matroids}.
\end{proof}

\subsection{Constructing codes with larger minimum distance}
\label{ss:larger-minimum-distance}

Recall that an algebro-geometric code $\calC$ is determined by a triple $(X,\calP,V)$ as in Section~\ref{ss:ag-codes}. It is interesting to study how a code changes when we vary these triples in a reasonable family. In Section~\ref{ss:set-up}, we chose $X = \Aff^1_x\times \Aff^1_y$, a collection of points $\calP$ broken up into $b$ batches of $r+1$ points, and the vector space $V$ in~\eqref{eq:vector-space-v}. Fix a nonnegative integer $z \in \ZZ_{\geq 0}$ with $0 \leq z \leq b-2$. In this subsection, we provide some numerical ruminations by altering the vector space $V$ to 
\[
V_z \coloneqq \left\{\sum_{\ell = 0}^{r-1} a_\ell(x) y^\ell : a_\ell(x)\in \FF_q[x], \deg a_\ell(x) \leq b-2-z \right\} \subset \FF_q[x,y].
\]
We denote the code arising from the construction in Section~\ref{ss:set-up} by
\[
\calC_z = \calC_z(X,\calP,V_z) \coloneqq \im\left(\ev_\calP\colon V_z \to \FF_q^n\right),
\]
where $n = \#\calP = b(r+1)$, as before. The codes from Section~\ref{ss:set-up} comprise the special 
case $z = 0$. 
The reader is invited to check that the proofs of Lemmas~\ref{lem:injectivity} and~\ref{lem:locality} go through in this new setting so that $\dim_{\FF_q} \calC_z = \dim_{\FF_q} V_z = (b-1-z)r$ and $\calC_z$ is an LRC with locality $r$. The Singleton-type bound for $\calC_z$ now gives the upper bound
\[
d_z \leq d_{\opt} \coloneqq (r + 1)(z + 1) + 2
\]
for the minimum distance of $\calC_z$. This bound increases with $z$, allowing for possible codes that have larger minimum distance than the ones we have studied so far.

\begin{example}
    \label{exam:rumination-1}
    Let $b=6$, $r=3$, and $q=31$. Then $n=6(3 + 1) = 24$, and we choose the set of points $\calP$ in six batches of four points as follows:
    \begin{align*}
        \mathcal{P}= \{&(1,1), (1,2), (1,3), (1,4), \\
        &(2,6), (2,7), (2,8),(2,9), \\
        &(3,11), (3,12), (3,13), (3,14), \\
        &(4,16), (4,17), (4,18), (4,19), \\
        &(5,21), (5,22), (5,23), (5,24), \\
        &  (6,25), (6,26), (6,27), (6,28) \}.
    \end{align*}
    A {\tt Magma} or {\tt SageMath} computation exhibits the following parameters for 
    the resulting code $\calC_z$ for $0\leq z\leq 3$:
    \begin{center}
    \begin{tabular}{c c c c}
        \hline
            $z$ & $[n,k,d]_q$ & $d_z$ & $d_{\opt}$ \\
            \hline
            0 & $[24, 15, 6]_{31}$ & $6$ & $6$  \\
            1 & $[24, 12, 9]_{31}$ & $9$ & $10$  \\
            2 & $[24, 9, 12]_{31}$ & $12$ & $14$ \\
            3 & $[24, 6, 16]_{31}$ & $16$ & $18$ \\
            \hline
    \end{tabular}
    \end{center}

    When $z=0$, we obtain an optimal code with minimum distance $6$, as expected in light of 
    Theorem~\ref{theo:line-fiber-optimal}.  We observe that as $z$ increases, the dimension 
    of the code decreases and its minimum distance increases. However, the Singleton-type 
    bound also increases, and these codes are not optimal when $1\leq z\leq 3$.
\end{example}

\begin{example}
    \label{exam:rumination-2}
    Let $b=10$, $r=2$, and $q = 37$. Then $n = 10(2 + 1) =30$, and we choose the set of points $\calP$ in ten batches of three points as follows:
    \begin{align*}
        \mathcal{P}=\{&(1,1),(1,2),(1,3),\\
        &(2,4), (2,5), (2,6), \\
        &(3,7), (3,8), (3,9), \\
        &(4,10), (4,11), (4,12), \\
        &(5,13), (5,14), (5,15)\\
        & (6,16), (6,17), (6,18), \\
        &(7,20), (7,21), (7,22), \\
        &(8,26), (8,27), (8,28), \\
        &(9,32), (9,33), (9,34),\\
        & (10,35), (10,36), (10,37)\}.
    \end{align*}
    A {\tt Magma} or {\tt SageMath} computation exhibits the following parameters for 
    the resulting code $\calC_z$ for $0\leq z\leq 7$:
    \begin{center}
        \begin{tabular}{c c c c}
           \hline
              $z$ & $[n,k,d]_q$ & $d_z$ & $d_{\opt}$ \\
              \hline
              0 & $[30, 18, 5]_{37}$ & $5$ & $5$  \\
              1 & $[30, 16, 8]_{37}$ & $8$ & $8$  \\
              2 & $[30, 14, 10]_{37}$ & $10$ & $11$ \\
              3 & $[30, 12, 12]_{37}$ & $12$ & $14$ \\
              4 & $[30, 10, 14]_{37}$ & $14$ & $17$ \\
              5 & $[30, 8, 17]_{37}$ & $17$ & $20$ \\
              6 & $[30, 6, 20]_{37}$ & $20$ & $23$ \\
              7 & $[30, 4, 23]_{37}$ & $23$ & $26$ \\
              \hline
        \end{tabular}
    \end{center}

    Observe that for all values of $z \geq 1$, the codes have a minimum distance greater than $r+3=5$, but for $2\leq z\leq 7$, the codes are not optimal.
\end{example}

In all cases above, the defect between the Singleton-type bound $d_{\opt}$ and the minimum 
distance of $\calC_z$ is small relative to $d_{\opt}$. In this sense, Examples~\ref{exam:rumination-1} 
and~\ref{exam:rumination-2} tantalizingly suggest that, even if the LRCs $\calC_z$ are not 
optimal, they are not too far from optimal. It would be interesting to further study these codes.

\section{Exploring nonoptimal codes}\label{sec:non-optimal}

In this section we build on Example~\ref{ex:badcodes} to showcase nonoptimal codes arising from the construction in Section~\ref{ss:set-up} when the alphabet size $q$ is small, and the locality $r$ exceeds $3$.

Fix fibers $b = 4$, locality $r = 4$, and an alphabet of size $q = 37$. Then the code $\mathcal{C}$ constructed in Section~\ref{ss:set-up} has length $n = b(r+1) = 20$; it arises by evaluating a set of points $\mathcal{P} = \{(x_i,y_{i,j})\}_{1\leq i\leq 4}^{1\leq j\leq 5}$ on a vector space of polynomials $V$ of dimension $k = (b-1)r = 12$. Since the exact values $x_2, x_3, x_4$ of the points on the base of the fibers do not affect the optimality of $\calC$ by Theorem~\ref{theo:IsomMatroids}, we now fix three distinct $x_2$, $x_3$, and $x_4$ in $\FF_{37}$ for the rest of the analysis.

\begin{figure}[ht]
\begin{center}
\begin{tikzpicture}
\draw[->][very thick] (0,0) -- (5.8,0);
\draw[->][very thick] (0,0) -- (0,6.4);
\node [label={[xshift=0.25cm, yshift=-0.35cm]$x$}] at (5.8,0){};
\node [label={[xshift=0.cm, yshift=-0.15cm]$y$}] at (0,6.4){};

\filldraw [thick, fill=blue!60!cyan,opacity=0.99](1.1,0.3) circle [radius=2pt];
\filldraw [thick, fill=blue!60!cyan,opacity=0.99](1.1,0.6) circle [radius=2pt];
\filldraw [thick, fill=blue!60!cyan,opacity=0.99](1.1,0.9) circle [radius=2pt];
\filldraw [thick, fill=blue!60!cyan,opacity=0.99](1.1,1.2) circle [radius=2pt];
\filldraw [thick, fill=blue!60!cyan,opacity=0.99](1.1,1.5) circle [radius=2pt];

\filldraw [thick, fill=blue!60!cyan,opacity=0.99](2.2,1.8) circle [radius=2pt];
\filldraw [thick, fill=blue!60!cyan,opacity=0.99](2.2,2.1) circle [radius=2pt];
\filldraw [thick, fill=blue!60!cyan,opacity=0.99](2.2,2.4) circle [radius=2pt];
\filldraw [thick, fill=blue!60!cyan,opacity=0.99](2.2,2.7) circle [radius=2pt];
\filldraw [thick, fill=blue!60!cyan,opacity=0.99](2.2,3) circle [radius=2pt];

\filldraw [thick, fill=blue!60!cyan,opacity=0.99](3.3,3.3) circle [radius=2pt];
\filldraw [thick, fill=blue!60!cyan,opacity=0.99](3.3,3.6) circle [radius=2pt];
\filldraw [thick, fill=blue!60!cyan,opacity=0.99](3.3,3.9) circle [radius=2pt];
\filldraw [thick, fill=blue!60!cyan,opacity=0.99](3.3,4.2) circle [radius=2pt];
\filldraw [thick, fill=blue!60!cyan,opacity=0.99](3.3,4.5) circle [radius=2pt];

\filldraw [thick, fill=blue!60!cyan,opacity=0.99](4.5,4.8) circle [radius=2pt];
\filldraw [thick, fill=blue!60!cyan,opacity=0.99](4.5,5.1) circle [radius=2pt];
\filldraw [thick, fill=blue!60!cyan,opacity=0.99](4.5,5.4) circle [radius=2pt];
\filldraw [thick, fill=blue!60!cyan,opacity=0.99](4.5,5.7) circle [radius=2pt];
\filldraw [thick, fill=blue!60!cyan,opacity=0.99](4.5,6) circle [radius=2pt];

\draw[very thick] (1.1,-0.05) -- (1.1,0.05);
\draw[very thick] (2.2,-0.05) -- (2.2,0.05);
\draw[very thick] (3.3,-0.05) -- (3.3,0.05);
\draw[very thick] (4.5,-0.05) -- (4.5,0.05);

\begin{scope}
    \node [label={[xshift=0cm, yshift=-.6cm]$x_{1}$}] at (1.1,0){};
    \node [label={[xshift=0cm, yshift=-.6cm]$x_{2}$}] at (2.2,0){};
    \node [label={[xshift=0.05cm, yshift=-.6cm]$x_{3}$}] at (3.3,0){};
    \node [label={[xshift=0.05cm, yshift=-.6cm]$x_{4}$}] at (4.5,0){};
\end{scope}

\draw[thick][-] (1.1,6.1) .. controls (1.1,3) and (1.1,1) .. (1.1,0.3); 
\draw[thick][-] (1.1,0.3) .. controls (1.3,-0.2) and (1.48,0.8) .. (1.5,0.9);	
\draw[thick][-] (1.5,0.9) .. controls (1.6,1.2) and (1.8,2) .. (2.2,2.1);	
\draw[thick][-] (2.2,2.1) .. controls (3.8,2.5) and (4,3.9) .. (3.3,3.9); 
\draw[thick][-] (3.3,3.9) .. controls (2.1,4) and (3.8,5.7) .. (4.5,5.1);	
\draw[thick][-] (4.5,5.1) .. controls (5.1,4.8) and (4.7,4.2) .. (4.6,3.9);
\draw[thick][-] (4.6,3.9) .. controls (4,1.7) and (4.4,0.8) .. (5.4,0.3);

\end{tikzpicture}
\caption{Polynomial vanishing in the first batch.}
\label{fig:poly-first-batch-zero}
\end{center}
\end{figure}
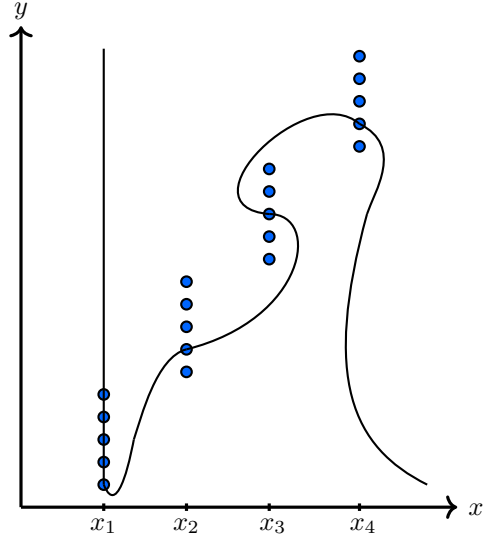

Example~\ref{ex:badcodes} suggests that if we want to find a nonzero code 
word $\mathbf{c} = \ev_\calP(f(x,y))$ in $\calC$ of length $\leq r + 2 = 6$, then we should take
\[
f(x,y) = (x - x_1)\sum_{\ell = 0}^3 g_\ell(x) y^\ell,
\]
where $g_\ell(x)$ has degree $\leq 1$ for $0\leq \ell\leq 3$. Write
\[
g(x,y) \coloneqq \sum_{\ell = 0}^3 g_\ell(x) y^\ell = 
(a_0+a_1x) + (a_2+a_3x)y + (a_4+a_5x)y^2 + (a_6+a_7x)y^3.
\]
The polynomial $f(x,y)$ vanishes along the five points 
\[
(x_1,y_{1,1}),(x_1,y_{1,2}),(x_1,y_{1,3}),(x_1,y_{1,4}),(x_1,y_{1,5})
\]
of $\calP$, as illustrated in Figure~\ref{fig:poly-first-batch-zero}. We want the factor $g(x,y)$ to vanish along nine more points among the remaining fifteen elements of~$\calP$. Without loss of generality, say $g(x,y)$ vanishes at the nine points 
\[
(x_2,y_{2,1}), (x_2,y_{2,2}), (x_2,y_{2,3}), (x_3,y_{3,1}), (x_3,y_{3,2}), (x_3,y_{3,3}),
    (x_4,y_{4,1}), (x_4,y_{4,2}), (x_4,y_{4,3}).
\]

The vanishing of $g(x,y)$ at the nine points can be written in matrix form as (compare with~\eqref{eq:explicit-vdm}):
\begin{equation*}
    \begin{bmatrix}
        1 & x_2 & y_{2,1} & x_2y_{2,1} & y_{2,1}^2 & x_2y_{2,1}^2 & y_{2,1}^3 & x_2y_{2,1}^3\\
        1 & x_2 & y_{2,2} & x_2y_{2,2} & y_{2,2}^2 & x_2y_{2,2}^2 & y_{2,2}^3 & x_2y_{2,2}^3\\
        1 & x_2 & y_{2,3} & x_2y_{2,3} & y_{2,3}^2 & x_2y_{2,3}^2 & y_{2,3}^3 & x_2y_{2,3}^3\\
        1 & x_3 & y_{3,1} & x_3y_{3,1} & y_{3,1}^2 & x_3y_{3,1}^2 & y_{3,1}^3 & x_3y_{3,1}^3\\
        1 & x_3 & y_{3,2} & x_3y_{3,2} & y_{3,2}^2 & x_3y_{3,2}^2 & y_{3,2}^3 & x_3y_{3,2}^3\\
        1 & x_3 & y_{3,3} & x_3y_{3,3} & y_{3,3}^2 & x_3y_{3,3}^2 & y_{3,3}^3 & x_3y_{3,3}^3\\
        1 & x_4 & y_{4,1} & x_4y_{4,1} & y_{4,1}^2 & x_4y_{4,1}^2 & y_{4,1}^3 & x_4y_{4,1}^3\\
        1 & x_4 & y_{4,2} & x_4y_{4,2} & y_{4,2}^2 & x_4y_{4,2}^2 & y_{4,2}^3 & x_4y_{4,2}^3\\
        1 & x_4 & y_{4,3} & x_4y_{4,3} & y_{4,3}^2 & x_4y_{4,3}^2 & y_{4,3}^3 & x_4y_{4,3}^3
    \end{bmatrix}
    \begin{bmatrix}
        a_0\\
        a_1\\
        a_2\\
        a_3\\
        a_4\\
        a_5\\
        a_6\\
        a_7
    \end{bmatrix} = 
    \begin{bmatrix}
        0\\
        0\\
        0\\
        0\\
        0\\
        0\\
        0\\
        0\\
        0
    \end{bmatrix}.
\end{equation*}

Let $M_{i,j}$ be the $8 \times 8$ matrix obtained by deleting the $3i + j - 6$ 
row of the above $9 \times 8$ matrix for $2\leq i\leq 4$ and $1\leq j\leq 3$. 
For example, we can check that
\begin{align*}
    \det M_{2,1}=   &(x_2 - x_3)(x_2 - x_4) (x_3 - x_4)^2 \\
                    &(y_{2,2}-y_{2,3})(y_{3,1}-y_{3,2})(y_{3,1}-y_{3,3})(y_{3,2}-y_{3,3})
                    (y_{4,1}-y_{4,2})(y_{4,1}-y_{4,3})(y_{4,2}-y_{4,3})\\
                    &r_{2,1}(y_{2,1},y_{2,2},y_{2,3},y_{3,1},y_{3,2},y_{3,3},y_{4,1},y_{4,2},y_{4,3})
\end{align*}
where (omitting the variables to avoid clutter) 
\begin{align*}
    r_{2,1} &= y_{2,2}^2 y_{2,3}^2 \sum_{j=1}^3(y_{3,j} - y_{4,j}) 
                + (y_{2,2}^2 + y_{2,3}^2) (y_{3,1}y_{3,2}y_{3,3}-y_{4,1}y_{4,2}y_{4,3})\\
                &-(y_{2,2}^2 y_{2,3}+y_{2,2} y_{2,3}^2) \sum_{1 \leq i < j \leq 3} (y_{3,i}y_{3,j} - y_{4,i}y_{4,j}) \\
                &+ y_{2,2} y_{2,3} \left( y_{3,1}y_{3,2}y_{3,3} - y_{4,1}y_{4,2}y_{4,3} + \sum_{\substack{1 \leq i < j \leq 3 \\ 1 \leq k \leq 3}} (y_{3,i}y_{3,j}y_{4,k} - y_{3,k}y_{4,i}y_{4,j}) \right)\\
                &-(y_{2,2}+y_{2,3}) \sum_{j=1}^3 (y_{3,1}y_{3,2}y_{3,3} y_{4,j} - y_{3,j}y_{4,1}y_{4,2} y_{4,3} ) \\
                &+y_{3,1}y_{3,2}y_{3,3}\sum_{1 \leq i < j \leq 3}y_{4,i}y_{4,j}  - y_{4,1}y_{4,2}y_{4,3}\sum_{1 \leq i < j \leq 3} y_{3,i}y_{3,j}.
\end{align*}
In general (omitting again some variables) 
\[
\det M_{i,j} = \left(\prod_{2 \leq \ell < m \leq 4 } (x_{\ell} - x_{m}) \right) (x_{i_1}-x_{i_2}) \left(\prod_{2 \leq u\leq 4} \prod_{1 \leq v < w \leq 3} (y_{u,v}-y_{u,w})\right) \frac{r_{i,j}(y_{2,1},\dots,y_{4,3})}{(y_{i,j}-y_{i,j_1}) (y_{i,j} - y_{i,j_2})},
\]
where $i_1 < i_2$; $i_1,i_2 \not = i$; $j_1 < j_2$; $j_1,j_2 \not = j$; and 
$r_{i,j}\in \FF_q[y_{2,1},y_{2,2},y_{2,3},y_{3,1},y_{3,2},y_{3,3},y_{4,1},y_{4,2},y_{4,3}]$ are 
homogeneous polynomials of degree $5$.

Since the $x_i$ are distinct and the $y_{u,v}$ are distinct for all $1\leq i,u\leq 4$ and all $1\leq v\leq 4$, the $9\times 8$ matrix above is singular precisely when the $9$ polynomials $r_{i,j}$ with $2\leq i\leq 4$ and $1\leq j\leq 3$ 
simultaneously vanish. Since the polynomials are homogeneous, their simultaneous vanishing defines a projective variety
\[
Z \coloneqq \{r_{2,1} = \cdots = r_{4,3} = 0\} \subset \PP^8_{\FF_{37}}.
\]
Rational points on $Z$ will now give rise to nonoptimal codes $\calC$. The variety $Z$ has dimension $6$, 
so one can improve the estimate~\eqref{eq:prob_estimate} to $A'q^6/q^8 = A'/q^2$, where $A'$ is the Lang--Weil constant for $Z$. For the convenience of the reader, we present the above calculations using \texttt{Magma} and \texttt{SageMath} in~\cite{AAAORAVA24}.

\begin{example}
    The point
    \[
    [ 17, 34, 14, 11, 8, 2, 36, 19, 1 ]\in \PP^8_{\FF_{37}}
    \]
    lies on the variety $Z$. We use this point to construct the set
    \[
    \begin{split}
    \calP = \{  &(4,3),(4,7),(4,28),(4,12),(4,21), \\
                &(9,\textbf{17}),(9,\textbf{34}),(9,\textbf{14}),(9,13),(9,22), \\
                &(16,\textbf{11}),(16,\textbf{8}),(16,\textbf{2}),(16,16),(16,23), \\
                &(25,\textbf{36}),(25,\textbf{19}),(25,\textbf{1}),(25,15),(25,26)\} \subset \FF_{37}^2.
    \end{split}
    \]
    Interpolating the points $(9,17), (9,34), (9,14), (16,11), (16,8), (16,2), (25,36), (25,19), (25,1)$, we construct
    \[
    f(x,y) = (x - 4) \left((1 + 26x) + (19 + 33x)y + (25 + 7x)y^2 + (8 + 34x)y^3\right)
    \]
    and compute
    \[
    \ev_\calP(f(x,y)) = (0, 0, 0, 0, 0, 0, 0, 0, 25, 16, 0, 0, 0, 5, 6, 0, 0, 0, 8, 11).
    \]
    This is a code word of length $6 < 7 = r + 3$ in the code 
    $\calC \coloneqq \mathrm{im} \ev_\calP$, witnessing the nonoptimality of $\calC$.
\end{example}

\section{Codes from projective space bundles}\label{sec:construction}

In this section we establish a general framework encompassing the codes of Section~\ref{sec:line-baseline-codes}. Those can be recovered by setting $m = 1$, $t = 1$, $\alpha = b-2$, and $\beta = r-1$ in the construction that follows.

\begin{figure}[ht]
\begin{center}
\begin{tikzpicture}
\draw[->][very thick] (-7,-0.3) -- (5.1,-0.8);
\draw[->][very thick] (-7,-0.3) -- (-7,6.5);

\begin{scope}
    \node [label={[xshift=0.4cm, yshift=-0.4cm]$\mathbb{P}^{1}$}] at (5.1,-0.8){};
    \node [label={[xshift=0.1cm, yshift=-.1cm]$ \mathbb{P}^{m}$}] at (-7,6.5){};
\end{scope} 

\filldraw [thick, fill=blue!10!white,opacity=0.20]  (-6,-1) -- (-6,5) -- (-4.7,6.3) -- (-4.7,0.3)--cycle;
\filldraw [thick, fill=blue!10!white,opacity=0.20] (-3.5,-1.1) -- (-3.5,4.9) -- (-2.4,6.2) -- (-2.4,0.2)--cycle;
\filldraw [thick, fill=blue!10!white,opacity=0.20] (0.5,-1.2) -- (0.5,4.8) -- (1.51,6.03) -- (1.51,0.03)--cycle;
\filldraw [thick, fill=blue!10!white,opacity=0.20] (3,-1.3) -- (3,4.7) -- (3.99,5.99) -- (3.99,-0.01)--cycle;

\draw[very thick] (-5.36,-0.42) -- (-5.36,-0.32);
\draw[very thick] (-2.97,-0.52) -- (-2.97,-0.42);
\draw[very thick] (0.96,-0.69) -- (0.96,-0.59);
\draw[very thick] (3.42,-0.8) -- (3.42,-0.7);

\begin{scope}
    \node [label={[xshift=0cm, yshift=-.6cm]$x_{1}$}] at (-5.36,-0.41){};
    \node [label={[xshift=0.05cm, yshift=-.6cm]$x_{2}$}] at (-2.97,-0.47){};
    \node [label={[xshift=0.24cm, yshift=-.6cm]$x_{b-1}$}] at (0.96,-0.64){};
    \node [label={[xshift=0.05cm, yshift=-.6cm]$x_{b}$}] at (3.42,-0.75){};
    \node [label={[xshift=0cm, yshift=-.6cm]$\cdots$}] at (-1.3,-0.54){};
\end{scope}

\filldraw [thick, fill=blue!99!white,opacity=0.99](-5.36,5.2) circle [radius=2pt];
\filldraw [thick, fill=blue!99!white,opacity=0.99](-5.7,4) circle [radius=2pt];
\filldraw [thick, fill=blue!99!white,opacity=0.99](-4.9,4.4) circle [radius=2pt];
\filldraw [thick, fill=blue!99!white,opacity=0.99](-5.36,2.37) circle [radius=2pt];
\filldraw [thick, fill=blue!99!white,opacity=0.99](-4.9,1) circle [radius=2pt];
\filldraw [thick, fill=blue!99!white,opacity=0.99](-5.7,1.8) circle [radius=2pt];

\filldraw [thick, fill=blue!99!white,opacity=0.99](-3.1,4.5) circle [radius=2pt];
\filldraw [thick, fill=blue!99!white,opacity=0.99](-2.6,4.9) circle [radius=2pt];
\filldraw [thick, fill=blue!99!white,opacity=0.99](-2.66,2.37) circle [radius=2pt];
\filldraw [thick, fill=blue!99!white,opacity=0.99](-2.96,0.87) circle [radius=2pt];
\filldraw [thick, fill=blue!99!white,opacity=0.99](-2.9,0.5) circle [radius=2pt];
\filldraw [thick, fill=blue!99!white,opacity=0.99](-2.9,3) circle [radius=2pt];

\filldraw [thick, fill=blue!99!white,opacity=0.99](0.8,4.5) circle [radius=2pt];
\filldraw [thick, fill=blue!99!white,opacity=0.99](1.3,3.2) circle [radius=2pt];
\filldraw [thick, fill=blue!99!white,opacity=0.99](1.1,3.4) circle [radius=2pt];
\filldraw [thick, fill=blue!99!white,opacity=0.99](1.16,2.37) circle [radius=2pt];
\filldraw [thick, fill=blue!99!white,opacity=0.99](0.8,1.4) circle [radius=2pt];
\filldraw [thick, fill=blue!99!white,opacity=0.99](0.99,0.6) circle [radius=2pt];

\filldraw [thick, fill=blue!99!white,opacity=0.99](3.36,3.37) circle [radius=2pt];
\filldraw [thick, fill=blue!99!white,opacity=0.99](3.8,4) circle [radius=2pt];
\filldraw [thick, fill=blue!99!white,opacity=0.99](3.4,3.) circle [radius=2pt];
\filldraw [thick, fill=blue!99!white,opacity=0.99](3.4,2) circle [radius=2pt];
\filldraw [thick, fill=blue!99!white,opacity=0.99](3.35,0.3) circle [radius=2pt];
\filldraw [thick, fill=blue!99!white,opacity=0.99](3.7,0.6) circle [radius=2pt];

\draw [thick][decorate,decoration={brace,amplitude=10pt,mirror,raise=4pt},yshift=0pt]
(3.81,0.5) -- (3.81,4.1) node [black,midway,xshift=0.99cm] {\footnotesize $tr+1$};

\end{tikzpicture}
\caption{The $b(tr+1)$ points in $\mathbb{P}^{1}\times \mathbb{P}^{m}$.}
\label{fig:generalcode}
\end{center}
\end{figure}
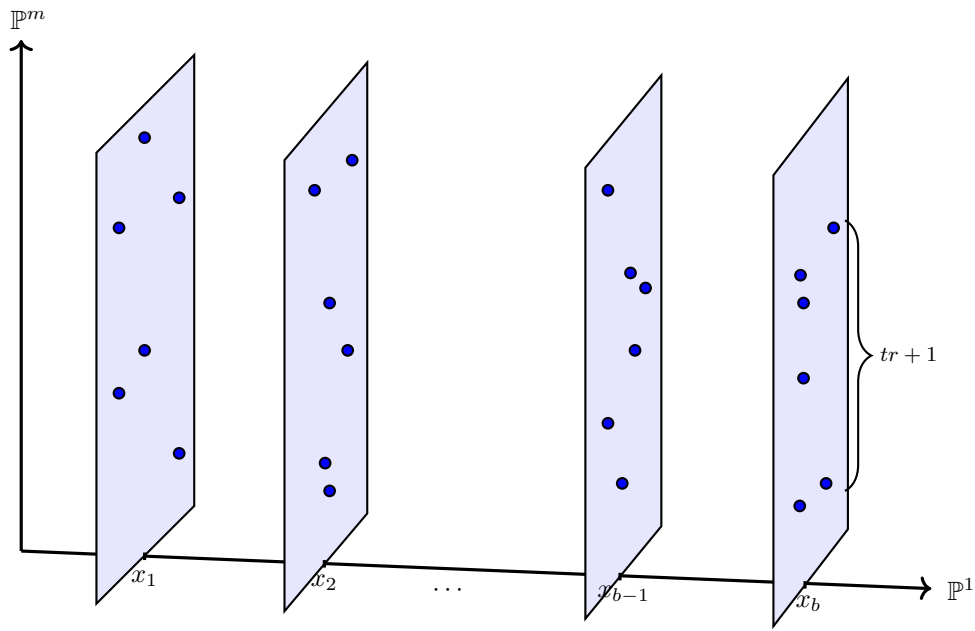

We fix positive integers $m,b,t, \alpha,\beta$ with $b \geq \alpha + 1$ and $q$ a prime power, and define 
\[
r\coloneqq \binom{\beta + m}{m} = \binom{\beta + m}{\beta}.
\]
Let $X \coloneqq \PP_x^1 \times \PP_{\mathbf{y}}^m$ and denote by $\pi_1\colon X\to \PP^1$ and $\pi_2\colon X\to \PP^m$ the projections onto each factor. Consider the vector space $V \coloneqq \Gamma(X,\OO_X(\alpha,\beta))$, which can be identified as the vector space of bi-homogeneous polynomials of bi-degree $(\alpha,\beta)$ with coefficients in $\mathbb{F}_q$. Note that
\[
\dim_{\mathbb{F}_q} V = (\alpha+1) \binom{\beta + m}{m} = (\alpha+1) \binom{\beta + m}{\beta} = (\alpha + 1)r,
\]
Pick $b$ points $x_1,\dots,x_b$ in $\PP^1$, and for each $1\leq i\leq b$ pick $t r + 1$ distinct points in the fiber $\pi_1^{-1}(x_i)$. For $1\leq i\leq b$ we denote by $A_i$ the batch of points picked in the fiber $\pi_1^{-1}(x_i)$. Furthermore, we will assume that the points of each fiber are in \emph{general position} inside $\PP^m$. This means that if $g(\mathbf{y})$ is a homogeneous polynomial of degree $\beta$ in $m+1$ variables, and $\mathbf{z}_1,\dots,\mathbf{z}_{r}\in A_i$ for some $i\in\{1,\dots,b\}$, then there exists at least one $\ell\in\{1,\dots,r\}$ such that $g(\mathbf{z}_{\ell}) \neq 0$. The disjoint union of these batches gives the set
\begin{equation*}
    \mathcal{P} = \coprod_{i=1}^{b}{A_i} = \{ (x_i,\mathbf{y}_{i,j}) \}^{1\leq j\leq tr+1}_{1\leq i\leq b}
\end{equation*}
whose points can be labeled as $P_1,\dots P_n$ because its cardinality is
\begin{equation*}
    n \coloneqq \#\mathcal{P} = b(tr+1).
\end{equation*}
Figure~\ref{fig:generalcode} illustrates this set-up geometrically. 

\begin{remark}
The condition of having all the points within each fiber in general position is very mild, it essentially says that we are using all the space available to us within the variety $\PP^m$.
\end{remark}

Let $\calC_m$ be the image of the evaluation map:
\begin{equation*}
\begin{tikzcd}[row sep=3em]
\ev_\mathcal{P} \colon &[-3.5em] V \arrow{r} & \FF_q^n,\\[-3.1em]
{} & f(x,y) \arrow[maps to]{r} & (f(P_1),\dots,f(P_n)).
\end{tikzcd}
\end{equation*}
The code $\mathcal{C}_m$ has code length $n$. We will show that $\mathcal{C}_m$ has dimension
\begin{equation*}
    k = (\alpha + 1) r,
\end{equation*}
and minimum distance
\[
    d \geq (b-\alpha) \left( (t-1) r + 2\right).
\]
In addition, we will show that $\mathcal{C}_m$ is locally recoverable, with locality $r$ and availability $t$. Let $\mathscr{C} = \{\calC_m\}_{m=1}^\infty$ be the family of codes parametrized by $m$; in this family, all codes have the same parameters $b$, $t$, $\alpha$, $\beta$, and $q$. Once the above claims about the code parameters of $\calC_m$ are established, it is straightforward to show that these codes are asymptotically good.

\begin{theorem}\label{theo:asymptotically-good-codes}
The family of codes $\mathscr{C}$ is asymptotically good.
\end{theorem}

\begin{proof}
    Note that with $b$, $t$, $\alpha$, and $\beta$ fixed, then
    \[
    \lim_{n \to \infty}\frac{d(\calC_m)}{n(\calC_m)} = \lim_{m \to \infty}\frac{d(\calC_m)}{n(\calC_m)} = \lim_{m\to \infty}\frac{(b-\alpha) \left( (t-1) r + 2\right)}{b \left(t r + 1 \right)} = \frac{(b-\alpha) (t-1)}{b t} > 0,
    \]
    and
    \[
    \lim_{n \to \infty}\frac{k(\calC_m)}{n(\calC_m)} 
    = \lim_{m \to \infty}\frac{k(\calC_m)}{n(\calC_m)} 
    = \lim_{n\to \infty}\frac{(\alpha + 1) r}{b \left(t r + 1 \right)} 
    = \frac{\alpha + 1}{b t} > 0. 
    \eqno\qed
    \]
\hideqed
\end{proof}

To show that the evaluation map $\ev_\calP \colon V \to \FF_q^n$ is injective, we begin with two auxiliary results. In plain terms, Lemma~\ref{lemma:fiber-vanishing} says that if a polynomial is zero on all the points of a fiber of our code then the polynomial restricted to that fiber is identically zero, and Lemma~\ref{lem:alpha-fiber-vanishing} says that a polynomial cannot be zero in more than $\alpha$ fibers; compare this result with Lemma~\ref{lem:max_number_zeros}.

\begin{lemma}\label{lemma:fiber-vanishing}
    Let $f(x,\mathbf{y}) \in V$ and fix $i\in \{1,\dots,b\}$. If $f(x_i,\mathbf{y}_{i,j}) = 0$ for all $\mathbf{y}_{i,j} \in A_i$, then $f(x_i,\mathbf{y})$ is the zero polynomial in $\mathbb{F}_q[y_0,\dots,y_m]$.
\end{lemma}

\begin{proof}
    Let $f(x,\mathbf{y}) = \sum_{\#I = \beta}{a_{I}(x)\mathbf{y}^I}$ be a polynomial in $V$. Since $f(x_i,\mathbf{y}_{i,j}) = 0$ for all $\mathbf{y}_{i,j} \in A_i$, the homogeneous polynomial $f(x_i,\mathbf{y}) \in \mathbb{F}_q[y_0,\dots,y_m]$ is of degree $\beta$ and has at least $\#A_i = tr + 1$ zeros in general position. Since this is more than $r - 1$ zeros, the polynomial $f(x_i,\mathbf{y})$ is identically zero, by definition of general position.
\end{proof}

\begin{lemma}[Generalized Fiber Vanishing Lemma] \label{lem:alpha-fiber-vanishing}
    Let $f(x,\mathbf{y}) \in V$ be a nonzero polynomial. Then $f(x_i,\mathbf{y}_{i,j}) = 0$ for all $\mathbf{y}_{i,j} \in A_i$ for $\leq \alpha$ values amongst $x_1,\dots,x_{b}$.
\end{lemma}

\begin{proof}
    Let $f(x,\mathbf{y}) = \sum_{\ell = 0}^{\alpha} F_{\ell}(\mathbf{y}) x^{\ell}$ be a nonzero polynomial in $V$. We will prove the result by contrapositive. Assume that $f(x,\mathbf{y})$ vanishes in $\geq \alpha+1$ fibers. Without loss of generality, assume that $\alpha+1$ of those fibers are $x_1,\dots,x_{\alpha+1}$. Then $f(x_i,\mathbf{y}) \equiv 0$ in $V$ for all $1\leq i\leq \alpha+1$ by Lemma~\ref{lemma:fiber-vanishing}. Consequently $f(x_i,\mathbf{y}_{i',j}) = 0$ for all $i,i' \in \{1,\dots,\alpha + 1\}$ and all $\mathbf{y}_{i',j}\in A_{i'}$; this is a higher-dimensional analogue of Observation~\ref{obs:too-many-zeros}. Thus the following equality holds for all $1\leq i\leq \alpha + 1$ and all $\mathbf{y}_{i,j_1},\dots, \mathbf{y}_{i,j_{\alpha + 1}} \in A_i$.
    \begin{equation*}
    \label{eq:vandermonde-fiber-vanishing}
    \begin{bmatrix}
        1 & x_1 & x_1^2 & \dots & x_1^{\alpha}\\
        1 & x_2 & x_2^2 & \dots & x_2^{\alpha}\\
        \vdots & \vdots & \vdots & \ddots & \vdots\\
        1 & x_{\alpha+1} & x_{\alpha+1}^2 & \dots & x_{\alpha+1}^{\alpha}
    \end{bmatrix}\begin{bmatrix}
        F_0(\mathbf{y}_{i,j_1})\\
        F_1(\mathbf{y}_{i,j_2})\\
        \vdots\\
        F_{\alpha}(\mathbf{y}_{i,j_{\alpha + 1}})
    \end{bmatrix} = \begin{bmatrix}
        0\\
        0\\
        \vdots\\
        0
    \end{bmatrix}
    \end{equation*}
Since the leftmost matrix is invertible, we must have
\begin{equation*}
    \begin{bmatrix}
        F_0(\mathbf{y}_{i,j_1})\\
        F_1(\mathbf{y}_{i,j_2})\\
        \vdots\\
        F_{\alpha}(\mathbf{y}_{i,j_{\alpha + 1}})
    \end{bmatrix} = \begin{bmatrix}
        0\\
        0\\
        \vdots\\
        0
    \end{bmatrix}.
\end{equation*}
Thus, for all $0\leq \ell\leq \alpha$, all $1\leq i\leq \alpha + 1$, and all $\mathbf{y}_{i,j}\in A_i$ 
we have $F_{\ell}(\mathbf{y}_{i,j}) = 0$. Since $F_{\ell}(\mathbf{y})$ is a homogeneous polynomial of 
degree $\beta$ in $m+1$ variables vanishing at $\#A_1+\cdots+\#A_{\alpha+1} = (\alpha+1)(tr+1)$ points in general 
position, we must have $F_{\ell}(\mathbf{y}) \equiv 0$. This implies that 
$f(x,\mathbf{y}) \equiv 0$ in $V$, a contradiction.
\end{proof}

\begin{corollary}\label{cor:general-injectivity}
The map $\ev_\calP \colon V \to \FF_q^n$ is injective. In particular, the code $\calC_m$ has dimension $k = (\alpha + 1) r$.
\end{corollary}

\begin{proof}
Let $f(x,\mathbf{y}) \in \ker(\ev_\mathcal{P})$. Then $f(x_i,\mathbf{y}_{i,j}) = 0$ for all $1\leq i\leq b$ and all $\mathbf{y}_{i,j}\in A_i$. Thus $f(x,\mathbf{y}) \equiv 0$ in $V$ by Lemma~\ref{lem:alpha-fiber-vanishing}.
\end{proof}

\begin{lemma}
\label{lem:general-injectivity}
The code $\mathcal{C}_m$ has locality $r$ and availability $t$.
\end{lemma}

\begin{proof}
Let $\mathbf{c} \in \mathcal{C}_m$ be a code word. Let $f(x,\mathbf{y}) = \sum_{\ell = 0}^{\alpha} F_{\ell}(\mathbf{y}) x^{\ell}$ be the polynomial in $V$ such that $\mathbf{c} = (f(P))_{ P \in \calP}$. Suppose that $\mathbf{c}$ is missing a symbol $c$; we may assume without loss of generality that $c = c_1$ is the evaluation at the point $(x_{1},\mathbf{y}_{1,1}) \in A_{1}$. Since $f(x_{1},\mathbf{y}) = \sum_{\ell = 0}^{\alpha} F_{\ell}(\mathbf{y}) x_{1}^{\ell}$ is a homogeneous polynomial of degree $\beta$ in $m+1$ variables, viewing its $\binom{\beta + m}{\beta}$ coefficients as unknowns we can set up the following overdetermined consistent system of equations
\begin{equation*}
\begin{cases}
\sum_{\ell = 0}^{\alpha} F_{\ell}(\mathbf{y}_{1,2}) x_{1}^{\ell} = f(x_1,\mathbf{y}_{1,2}) = c_2\\
\sum_{\ell = 0}^{\alpha} F_{\ell}(\mathbf{y}_{1,3}) x_{1}^{\ell} = f(x_1,\mathbf{y}_{1,3}) = c_3\\
\vdots\\
\sum_{\ell = 0}^{\alpha} F_{\ell}(\mathbf{y}_{1,tr}) x_{1}^{\ell} = f(x_1,\mathbf{y}_{1,tr}) = c_{tr}\\
\sum_{\ell = 0}^{\alpha} F_{\ell}(\mathbf{y}_{1,tr+1}) x_{1}^{\ell}  = f(x_1,\mathbf{y}_{1,tr+1}) = c_{tr+1}
\end{cases}
\end{equation*}
where $c_{j}$ for $j\in\{2,\dots,tr+1\}$ is the known value of $f(x,\mathbf{y})$ evaluated at $(x_{1},\mathbf{y}_{i,j}) \in A_1$. Since the points in $A_{1}$ are in general position, any choice of $r$ of the above equations suffices to solve the system, unequivocally determining the polynomial $f(x_{1},\mathbf{y})$. The missing symbol can then be recovered by evaluating $c_1 = f(x_1,\mathbf{y}_{1,1})$. Finally, note that there are $t$ disjoint sets of $r$ equations determining the polynomial $f(x_{1},\mathbf{y})$.
\end{proof}

\begin{proposition}
\label{prop:bound-minimal-distance}
    The minimum distance $d$ of the code $\mathcal{C}_m$ satisfies
    \begin{equation*}
        d \geq (b-\alpha) \left( (t-1) r + 2\right).
    \end{equation*}
\end{proposition}

\begin{proof}
Let $f(x,\mathbf{y})$ be a nonzero polynomial in $V$. Then $f(x_i,\mathbf{y}_{i,j}) = 0$ for all $\mathbf{y}_{i,j}\in A_{i}$ for $s \leq \alpha$ values of $x_1,\dots,x_b$, by Lemma~\ref{lem:alpha-fiber-vanishing}. If $i\in \{1,\dots,b\}$ is such that $f(x_i,\mathbf{y}_{i,j}) \neq 0$ for at least one $\mathbf{y}_{i,j}\in A_{i}$, then there are $\leq r-1$ points $\mathbf{y}_{i,j_{1}},\dots,\mathbf{y}_{i,j_{r-1}}\in A_{i}$ such that $f(x_i,\mathbf{y}_{i,j_{1}}) = \cdots = f(x_i,\mathbf{y}_{i,j_{r-1}}) = 0$ because the points in $A_i$ are in general position. 
Consequently the number of zeros of $f(x,\mathbf{y})$ along each of the $b - s$ non zero-fibers is $\geq (tr + 1) - (r-1)$, and thus the weight $\omega$ of the code word $\mathbf{c} \coloneqq \ev_{\calP}(f(x,\mathbf{y}))$ is bounded by
\[
\omega \geq (b-s)\left( (t r + 1) - \left(r - 1\right) \right) = (b-s) \left( (t-1) r + 2\right).
\]
As a function of $s$, the right hand side of the above equality is minimized when $s$ is maximized. Hence
\begin{equation*}
d \geq (b-\alpha) \left( (t-1) r + 2\right).
\eqno\qed
\end{equation*}
\hideqed
\end{proof}

We now have all the ingredients to prove Theorem~\ref{theo:higher-dimensional-summary}.

\begin{proof}[Proof of Theorem~\ref{theo:higher-dimensional-summary}]
    This is a consequence of Corollary~\ref{cor:general-injectivity}, Lemma~\ref{lem:general-injectivity}, and Proposition~\ref{prop:bound-minimal-distance}. The asymptotic behavior of the family of codes $\mathscr{C} = \{\calC_m\}_{m=1}^\infty$ was studied in Theorem~\ref{theo:asymptotically-good-codes}.
\end{proof}

\subsection{Gilbert--Varshamov bounds}

Given an asymptotically good family of codes $\mathscr{C}$ like the one produced above, it is natural to ask if the limiting relative parameters in the proof of Theorem~\ref{theo:asymptotically-good-codes} lie near a Gilbert--Varshamov-type bound for LRCs with availability. Barg, Tamo, and Frolov obtained a bound for LRCs with availability in~\cite{TamoBargFrolov16}*{Theorem~B}. It is difficult to derive an asymptotic bound from their formulas, although they do this successfully in the case of availability $t = 2$~\cite{TamoBargFrolov16}*{Figure~1}. Their asymptotic bounds, as $n \to \infty$, hold the locality $r$ fixed. Unfortunately, for the codes in our family, $n$~and $r$~are inextricably linked; if $n \to \infty$ then $r \to \infty$ as well. Therefore, our family of codes $\mathscr{C}$ is not amenable to an asymptotic GV-bound analysis. Exploring whether an alternative construction could be analyzed from this perspective would be worthwhile.

\section*{Acknowledgments}

This project began at the Latinx Mathematics Research Community, sponsored 
by the American Institute of Mathematics and the National Science Foundation. The 
authors thank their fellow LMRC community members Maurice Fabien, Zachary Flores, 
Therese-Marie Landry, Adriana Salerno, and Gustavo Terra Bastos for useful discussions. 
The continuation of this project was made possible by a SQuaRE at the American 
Institute of Mathematics. The authors thank the American Institute of Mathematics 
for providing a supportive and mathematically rich environment. 
The first author was supported by the National Science Foundation individual grant 
DMS-2316892 while working on this project. The second author 
was partially supported by the AMS-Simons Research Enhancement Grants for PUI 
Faculty while working on this project. 
The fourth author was supported by an AMS-Simons Travel Grant. 
The fourth, fifth, and sixth authors conducted some of this work supported by the National 
Science Foundation under Grant No.\ DMS-1928930, while in residence at the Simons Laufer 
Mathematical Sciences Research Institute (formerly MSRI) in Berkeley, California, during the 
summer of 2024, spring 2024, and the spring of 2023, respectively. 
The sixth author was also supported by the National Science Foundation individual grants 
Nos.\ DMS-1902274 and DMS-2302231 while working on this project. 
We used \texttt{Magma}~\cite{BosmaCannonPlayoust97} and \texttt{SageMath}~\cite{Sagemath24} 
for computations.

\bibliography{ref}

\end{document}